\documentclass[11pt]{article}
\usepackage{amsmath}
\usepackage{graphicx}
\usepackage{enumerate}
\usepackage{natbib}
\usepackage{url} 
\usepackage{multicol, multirow}
\usepackage[normalem]{ulem}

\newcommand{\blind}{1}

\usepackage{amsthm}
\theoremstyle{plain}
\newtheorem{thm}{Theorem}
\newtheorem{lem}[thm]{Lemma}

\usepackage{amssymb}
\def\scD{\mathcal{D}}
\def\scP{\mathcal{P}}
\def\scL{\mathcal{L}}
\newcommand{\expct}{\mathbb{E}}
\newcommand{\cov}{\mathbb{C}\mbox{ov}}
\newcommand{\kl}{\mathrm{KL}}
\newcommand{\gam}{\textit{Ga}}

\newcommand{\scB}{\mathcal{B}}
\newcommand{\scF}{\mathcal{F}}
\newcommand{\scT}{\mathcal{T}}
\newcommand{\scU}{\mathcal{U}}

\newcommand{\iid}{\stackrel{\mbox{\tiny IID}}{\sim}}
\newcommand{\bbR}{\mathbb{R}}
\def\bbP{\mathbb{P}}
\newcommand{\del}[1]{\frac{\partial}{\partial#1}}
\newcommand{\deltwo}[1]{\frac{\partial^2}{\partial#1^2}}
\newcommand{\deltwom}[2]{\frac{\partial^2}{\partial#1\partial#2}}
\newcommand{\bbB}{\mathbb{B}}
\newcommand{\bbH}{\mathbb{H}}
\newcommand{\bbC}{\mathbb{C}}

\newcommand{\Holder}{\mathcal{H}}

\def\lgp{\mathrm{LGP}}
\def\alphalo{\underline\alpha}
\def\alphahi{\overline\alpha}

\def\sigmalo{\underline\sigma}
\def\sigmahi{\overline\sigma}
\def\xilo{\underline\xi}

\def\norm{\mathrm{norm}}
\DeclareMathOperator*{\plim}{plim}
\def\diam{\mathrm{diam}}
\def\ti{\alpha_+}
\def\si{\zeta}

\def\Thetao{\mathring \Theta}
\newcommand{\mrow}[2]{\multirow{#1}{*}{#2}}
\newcommand{\nab}[1]{\nabla_{\kern -.2em#1}}

\def\taulo{\underline \tau}

\usepackage{xcolor}

\newcounter{lem5counter}

\begin{document}

\def\spacingset#1{\renewcommand{\baselinestretch}%
{#1}\small\normalsize} \spacingset{1}


\if1\blind
{
  \title{\bf Heavy-Tailed Density Estimation}
  \author{Surya T Tokdar\thanks{
    This research was partially supported by grants DMS1613173 and DMS2014861 from the National Science Foundation}\hspace{.2cm}\\
    Department of Statistical Science, Duke University\\
    and \\
    Sheng Jiang \\
    Department of Statistics, University of California, Santa Cruz\\
    and\\
    Erika L Cunningham \\
    Department of Statistical Science, Duke University}
    \date{}
  \maketitle
} \fi

\if0\blind
{
  \bigskip
  \bigskip
  \bigskip
  \begin{center}
    {\LARGE\bf Heavy-Tailed Density Estimation}
\end{center}
  \medskip
} \fi

\begin{abstract}
A novel statistical method is proposed and investigated for estimating a heavy tailed density under mild smoothness assumptions. Statistical analyses of heavy-tailed distributions are susceptible to the problem of sparse information in the tail of the distribution getting washed away by unrelated features of a hefty bulk. The proposed Bayesian method avoids this problem by incorporating smoothness and tail regularization through a carefully specified semiparametric prior distribution, and is able to consistently estimate both the density function and its tail index at near minimax optimal rates of contraction. A joint, likelihood driven estimation of the bulk and the tail is shown to help improve uncertainty assessment in estimating the tail index parameter and offer more accurate and reliable estimates of the high tail quantiles compared to thresholding methods.
\end{abstract}

\noindent%
{\it Keywords:}  Semiparametric estimation, logistic Gaussian processes, posterior contraction, tail index estimation, regular variation. 

\spacingset{1} 

\section{Introduction}
\label{chap:intro}
For a heavy-tailed density with subexponential tail decay, the exceedance probabilities of a sample sum and a sample maximum are of the same order. A random sample drawn from such a density is likely to contain a small fraction of extreme observations whose magnitudes overshadow the sum total of the remaining magnitudes. This property is expressive of many naturally occurring phenomena, e.g., precipitation \citep{katz2002statistics}, financial returns or insurance loss \citep{embrechts2013modelling}, and material or fatigue strength \citep{castillo2012extreme}
. 
However, statistically estimating a heavy tailed density 
from a random sample 
could be challenging if estimation was sought under only smoothness conditions. 
Two densities can be arbitrarily close in total variation distance while displaying entirely different tail decay rates. 
Estimation methods with rich shape flexibility and guaranteed $L^1$ estimation consistency may provide no meaningful inference on the tails of the distribution; see \cite{markovich2007nonparametric, li2019posterior} for detailed discussions and cautionary results on kernel mixture models.

When interest focuses on estimating only tail features, e.g., extrapolating to high quantiles from limited data, it is common to exclude all but the most extreme observations so that the tail speaks for itself.
The Pickands-Balkema-de Haan Theorem \citep{balkema1974residual, pickands1975statistical} justifies the so-called {\it peaks-over-threshold} estimation methods, where a generalized Pareto distribution (GPD) is fitted to the subsample of observations exceeding a high threshold; see \cite{de2010parameter} for a review. It also motivates nonparametric methods \citep{hill1975simple, pickands1975statistical, dekkers1989moment, alves2001location}  based on only high sample quantiles for estimating the asymptotic tail decay rate of {densities $f(y)$ whose survival function $\bar F(y) = \int_y^\infty f(t) dt$ is regularly varying, i.e., 
\begin{equation}
\bar F(y) = y^{-\alpha}L(y),~y > 0,
\label{eq:power law}
\end{equation}
for some $\alpha > 0$ where $L(y)$ is a slowly varying function, i.e., $\lim_{y \to \infty} L(ay)/L(y) = 1$ for every $a > 0$. We shall call such an $f(y)$ a regularly varying density with {\it tail index} $\alpha$, which may be recovered from $f$ as $\alpha = \ti(f):=-\lim_{y\to \infty}\frac{\log \bar F(y)}{\log y}$.
} 

A data driven threshold selection is critical to the analysis, but an optimal choice proves a steep challenge in practice. Diagnostic plots may point to multiple regimes of transition to the tail. Automatic threshold estimation methods gloss over such ambiguity with unverifiable tail assumptions and fail to account for the associated uncertainty in subsequent analyses \citep{scarrot2012review}. 
Several methods have been proposed to estimate the entire density function by splicing together a mixture model for the bulk with a GPD tail attachment \citep{tancredi2006accounting, macdonald2011flexible, do2012semiparametric}. Although, in theory, these methods partially account for threshold uncertainty, they employ heuristic estimation methods supported by little mathematical analysis.

Toward a more formal statistical methodology we consider the semiparametric model
\begin{equation}
f(y) = p_{\theta,\psi}(y) := g_\theta(y)\psi(G_\theta(y)), \quad y > 0,
\label{eq:model}
\end{equation}
where $g_\theta(y) = \sigma^{-1}\{1 + y/(\alpha\sigma)\}^{-(\alpha + 1)}$, $G_\theta(y) = \int_0^y g_\theta(z)dz$, $y > 0$, are the density and distribution functions of a GPD with location 0, scale $\sigma$ and shape $1/\alpha$; here $\theta = (\alpha, \sigma)\in (0,\infty)^2$ is an unknown vector, and, $\psi$ is an unknown density function on $(0,1)$. Under this model, $Y \sim f(y)$ if and only if $U := G_\theta(Y) \sim \psi$, and, $\ti(f) = \ti(g_\theta) = \alpha$ under a regularity condition on $\psi(u)$ as $u\to 1$ (Lemma \ref{lem:tail}). \cite{markovich2007nonparametric} offers a thorough analysis of an estimation approach where one first obtains an estimate $\hat\theta$ of $\theta$ by thresholding data $Y_1, \ldots, Y_n$ at a high quantile, and then a nonparametric estimate $\hat\psi$ of $\psi$ is obtained based on the transformed data $\hat U_i = G_{\hat\theta}(Y_i)$, $i = 1, \ldots, n$. With $\hat\psi$ estimated by a variable kernel mixture, the back-transformed density $\hat f = p_{\hat\theta, \hat\psi}$ offers optimal estimation of $f$ under the mean integrated square error loss. Such a two-stage approach does not account for threshold choice uncertainty in the estimation of $f$ or any subsequent analyses. 
It also fails to take advantage of the estimate of the bulk to improve tail estimation. 

We consider a likelihood-based alternative approach where $\theta$ and $\psi$ are jointly estimated under a Bayesian extension of \eqref{eq:model}. A Bayesian formulation immediately facilitates information sharing between the bulk and the tail and offers a joint assessment of uncertainty of the extreme and non-extreme features. But important new questions arise on both Bayesian and frequentist sides. What is a principled way to choose a prior distribution on the nonparametric density $\psi$? What are the statistical properties of the resulting estimates? These questions could be partially addressed by examining asymptotic concentration properties of the posterior distribution resulting from a specific prior allocation. 
%
%
We show that with a logistic Gaussian process (LGP) prior on $\psi$ \citep{leonard1978density, lenk1988logistic, lenk1991towards, tokdar2007towards}, the posterior distribution on $f$ given a random sample $Y_1,\ldots, Y_n$ from an $f^*$ concentrates around $f^*$ whenever the latter is continuous and regularly varying. Moreover, the posterior distributions on $f$ and $\ti(f)$ simultaneously concentrate around $f^*$ and $\ti(f^*)$ 
%
{at polynomially fast contraction rates that are nearly minimax optimal}, whenever $f^*=p_{\theta^*,\psi^*}$ with a  sufficiently smooth $\psi^*$ . 
It is significant that the LGP prior enables the likelihood function to preserve relevant information on tail quantities; no other example has been worked out before \citep{li2019posterior}. Moreover, guaranteeing posterior contraction across a large model subspace is tantamount to adopting the principle of intersubjective prior allocation to facilitate asymptotic merger of beliefs \citep{diaconis1986consistency}.


Computational details are provided for an efficient and streamlined implementation 
making it feasible to analyze data sets consisting of several thousand records. Finite sample properties are examined with an extensive simulation study which corroborates the asymptotic analysis result of accurate tail index estimation under strong GPD tail match, and complements it by revealing that even {under deviations} from a GPD tail, estimates of high tail quantiles are much superior compared to those obtained from thresholding methods. An analysis of daily precipitation records is presented to highlight potential benefits of the joint semiparametric estimation in mitigating ambiguity regarding threshold choice and providing tight but robust estimates of high tail quantiles.

\section{Estimation model}
\subsection{Tail index expression}
We restrict to the case where the support of $f$ is $[a,\infty)$ for a known finite number $a$, which is set to be zero without any loss of generality. The primary goal of the analysis is taken to be estimating the entire density $f$ accurately in $L^1$ or comparable metrics, while also accurately estimating its heavy right tail. Toward this, we first show that the GPD-transformation model \eqref{eq:model} is expressive of an entire range of polynomial tail decay rates under a regularity assumption on $\psi$. 

Let $\scP$ denote the class of densities $\psi$ on $(0,1)$ satisfying $\bar \Psi(1 - u) = u\tilde L(1/u)$ for some slowly varying function $\tilde L$; here $\Psi$ denotes the distribution function of $\psi$ and $\bar \Psi = 1 - \Psi$. Note that if $L(y)$ is slowly varying then
\begin{equation}
\lim_{y \to \infty} \frac{L(a(y)y)}{L(y)} = 1
\label{eq:slow}
\end{equation}
for any function $a(y)$ with a limit $a_\infty := \lim_{y \to \infty} a(y) \in (0,\infty)$.

\begin{lem}
\label{lem:tail}
If $\theta = (\alpha, \sigma) \in (0,\infty)^2$ and $\psi \in \scP$ then $f = p_{\theta, \psi}$ is regularly varying with tail index $\alpha$. Conversely, if $f$ is a regularly varying density on $(0,\infty)$ with tail index $\alpha > 0$ then for every $\sigma > 0$, $f = p_{(\alpha,\sigma), \psi}$ for some $\psi \in \scP$. 
 \end{lem}

\begin{proof}
If $f = p_{\theta, \psi}$ then $\bar F(y) = \bar \Psi(1 - \bar G_\theta(y)) = \bar G_\theta(y)\tilde L(1/\bar G_\theta(y))$, with $\bar G_\theta(y) = 1 - G_\theta(y) = y^{-\alpha}L_\theta(y)$, $L_\theta(y) = \{1/y + 1/(\alpha\sigma)\}^{-\alpha} \to c_\theta:=(\alpha\sigma)^\alpha$ as $y \to \infty$. Therefore, $\bar F(y) = y^{-\alpha}L(y)$ where $L(y) = L_\theta(y)\tilde L(y^\alpha/L_\theta(y))$ is slowly varying 
by \eqref{eq:slow}.
Conversely, if $f$ is a regularly varying density on $(0,\infty)$ with tail index $\alpha > 0$ and $\theta = (\alpha, \sigma)$ for some $\sigma > 0$, then $f = p_{\theta, \psi}$ where
\begin{equation}
\psi(u) = \frac{f(G_\theta^{-1}(u))}{g_\theta(G_\theta^{-1}(u))},~~u \in (0,1).
\label{eq:psi}
\end{equation}
It is trivial to check that $\psi$ is a density on $(0,1)$ with $\textstyle \bar \Psi(1 - u) = \bar F(\bar G_\theta^{-1}(u)) = \bar F(\alpha\sigma(u^{-\frac1\alpha} - 1)) = u\tilde L(\frac1u)$
where $\tilde L(y) = \alpha\sigma\{1 - 1/y^{1/\alpha}\}^{-\alpha}L(\alpha\sigma\{y^{1/\alpha} - 1\})$, with $L$ denoting the slowly varying component of $\bar F$. By \eqref{eq:slow}, $\tilde L$ itself is a slowly varying function.
\end{proof}





Lemma \ref{lem:tail} says, with $\psi \in \scP$ the semiparametric model \eqref{eq:model} is fully expressive of all regularly varying densities on $(0,\infty)$ with tail index uniquely identified by the model parameter $\alpha$. It also says that the pair $(\sigma, \psi)$ is not uniquely identifiable. 
Although one could fix $\sigma$ and have both $\alpha$ and $\psi$ uniquely identified under \eqref{eq:model}, no obvious choice presents itself. 
Instead, we find it more useful to retain the scale expressiveness of the model to adjust for implicit shape preferences of any nonparametric prior on $\psi$. The LGP prior introduced below concentrates around $\psi$ functions such that the derivatives of $\log \psi$ are small in magnitude; a bias toward smooth functions being critical to statistical regularization. A flexible choice of the pairing $\sigma$ creates an important counterbalance. It offers an entire arc of equivalent $(\sigma, \psi)$ pairs for a given $f$, increasing the possibility that at least some of these pairs will be favorable to LGP shape bias and hence will enjoy high posterior concentration. For example, when $f = g_{(2,1)}$ the pair $(\sigma=1, \psi\equiv1)$ presents a favorable representation. But if one now adds a little contamination a different pair with a $\sigma \ne 1$ could be more suitable, even if the contamination does not alter the tail behavior. Figure \ref{fig:scale} shows a concrete example with a gamma contamination. 

\begin{figure*}[t]
\centering
\includegraphics[width=\textwidth]{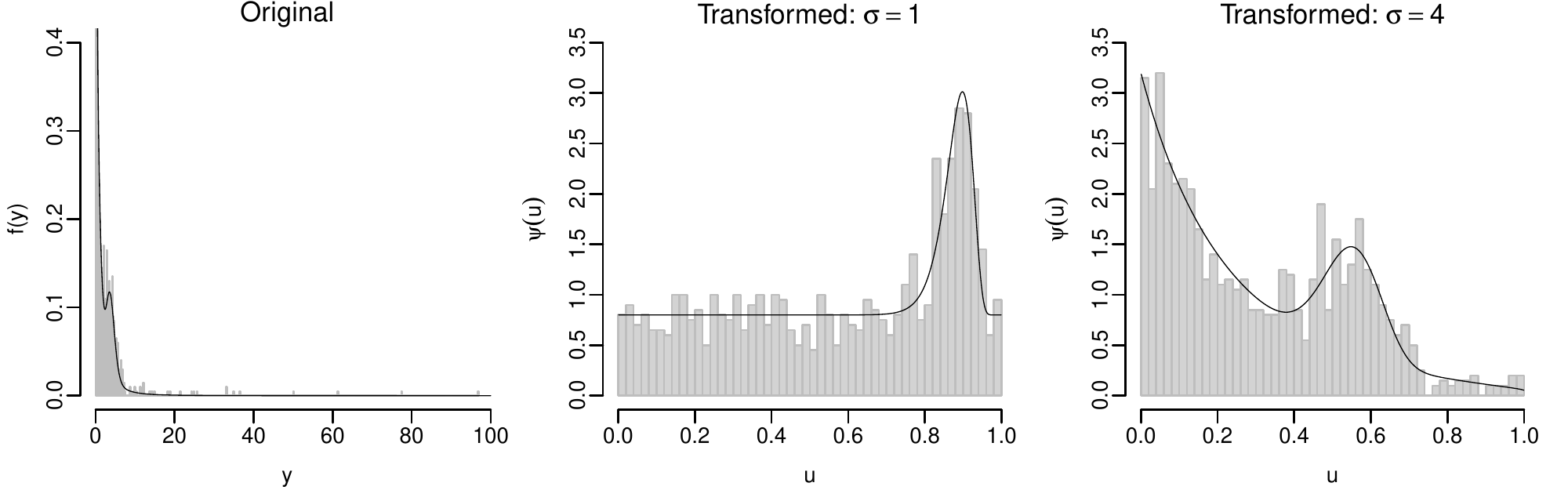}
\caption{Lack of identifiability of $(\sigma, \psi)$ and the importance of scale expressiveness. Left: graph of $f(y) = 0.8 \times g_{(2,1)}(y) + 0.2\times\tilde g(y)$ where $\tilde g(y)$ is the density of a gamma distribution with mean 4 and variance 1; overlaid on the histogram of a sample of size $n=1000$ drawn from the same. The tail of $f(y)$ is completely dominated by that of $g_{(2,1)}$. Remaining panels show graphs of $\psi(u)$ in \eqref{eq:psi} overlaid on the histogram of transformed data with $\alpha=2$ and two choices of the scale: $\sigma = 1$ (middle) and $\sigma = 4$ (right). The larger scale value produces a flatter $\psi$ which is more favorable to the LGP prior. For the data displayed here, the posterior concentrates around $(\alpha, \sigma) = (2.1, 4)$ with $[2.2,7.2]$ giving a 95\% interval for $\sigma$. 
}
\label{fig:scale}
\end{figure*}

Toward a Bayesian analysis, we choose a product prior distribution $\pi_\theta = \pi_\alpha\times \pi_\sigma$ on $\theta = (\alpha, \sigma)$, where $\pi_\sigma$ is the half-Cauchy distribution on $(0,\infty)$ and $\pi_\alpha$ is the distribution of $\alpha = \alphalo + (2 - \alphalo) \cdot e^{\zeta/1.5}$ with $\zeta$ distributed according to the standard logistic distribution on the real line. The choice of $\pi_\alpha$ restricts $\alpha > \alphalo$ with probability one, where $\alphalo > 0$ is treated as a hyperparameter to be fixed by the modeler. The numerical analyses presented in Section \ref{sec:finite-samp} were carried out with $\alphalo = 0.5$, for which the extreme value index $\xi = \frac1\alpha$ has unimodal density on $(0,2)$ with a gentle peak at $\xi = 0.5$ (i.e., $\alpha = 2$). We also experimented with $\alphalo = 0.1$ and all posterior estimates were found to be essentially the same as with $\alphalo=0.5$.


\subsection{LGP prior for $\psi$}\label{sec:lgpnonparprior}
Let $C[0,1]$ denote the space of real, continuous functions on $[0,1]$. For any $\omega \in C[0,1]$, its logistic transform $\scL(\omega)$, defined as
$(\scL\omega)(u) = \frac{e^{\omega(u)}}{\int_0^1 e^{\omega(t)}dt}, u \in (0,1),$
is a well defined probability density function on $(0,1)$. When $\omega$ is a Gaussian process with $\expct[\omega(u)] = \mu(u)$ and $\cov[\omega(u), \omega(v)] = c(u, v)$ such that $\omega \in C[0,1]$ with probability one, the probability law of the random density function $\scL(\omega)$ is called a logistic Gaussian process distribution, denoted $\lgp(\mu, c)$. 

We adopt a hierarchical LGP prior for $\psi$ in \eqref{eq:model}. Let $c_\lambda(u, v) = \exp\{-\lambda^2(u - v)^2\}$ denote the unit variance Gaussian covariance kernel with inverse length-scale parameter $\lambda > 0$. It is well known that if $\omega$ is a mean zero Gaussian process with covariance $\kappa^2 c_\lambda$ for some $\kappa > 0$, then $\omega \in C[0,1]$ with probability one, and hence, the probability distribution $\lgp(0, \kappa^2 c_\lambda)$ is well defined for every $\kappa > 0, \lambda > 0$. The prior on $\psi$ is implicitly defined by the hierarchy
\begin{equation}
\psi \sim \lgp(0, \kappa^2 c_\lambda),\quad (\kappa^2, \lambda ) \sim \pi_{\kappa^2} \times \pi_\lambda,
\label{eq:prior}
\end{equation}
with the distributions $\pi_{\kappa^2}$ and $\pi_\lambda$ on $(0,\infty)$ described below. 

It is clear that if $\omega \in C[0,1]$ and $\psi = \scL(\omega)$, then $\psi(1) := \lim_{u \to 1} \psi(u)$ exists and $\psi(1) \in (0,\infty)$. By mean value theorem $\bar \Psi(1 - u) = u\psi(t(u))$ for some $t(u) \in [1-u,1]$. Consequently, $\tilde L(y) = \psi(t(1/y))$, is slowly varying because  $\lim_{y \to \infty} \tilde L(y) = \psi(1) \in (0,\infty)$. Therefore, under the hierarchical LGP prior adopted here, $\Pr(\psi \in \scP) = 1$. Of course, the prior support of $\psi$ is actually smaller than $\scP$, because $\lim_{u \downarrow 0} u^{-1} \bar \Psi(1-u)  \in (0,\infty)$ almost surely under the prior, whereas $\scP$ contains densities $\psi$ where this limit may be zero, infinity or undefined. This may suggest that the induced prior distribution on $p_{\theta, \psi}$ may not have full support within the class of regularly varying densities. The theorem below reassures that no loss is incurred in a probabilistic sense. Below we assume $0 \le \alphalo < \alphahi \le \infty$ are such that $\alpha \in (\alphalo, \alphahi)$ with probability one under the prior $\pi_\theta$. 

\begin{thm}
\label{thm:kl}
Let $f^*$ be any bounded, continuous, regularly varying density on $(0,\infty)$ with tail index $\alpha^* \in (\alphalo, \alphahi)$. If $(\theta,\psi) \sim \pi_\theta\times \lgp(0, \kappa^2 c_\lambda)$ for some $\kappa > 0, \lambda > 0$, then for every $\epsilon > 0$,
$\Pr(d_{\kl}(f^*, p_{\theta,\psi}) < \epsilon) > 0,$
where $d_{\kl}(f,g) = \int f(y)\log\{{f(y)}/{g(y)}\}dy$ denotes the Kullback-Leibler divergence of $f$ from $g$. 
\end{thm}

\begin{proof}
Let $\epsilon > 0$ be given. Fix a $0 < \delta < 1 - e^{-\epsilon/2}$. Consider any $\theta_0 = (\alpha_0, \sigma_0)$ where $\alphalo < \alpha_0 < \alpha^*$ and $\sigma_0 > 0$. Let $\psi_0$ be defined as in \eqref{eq:psi} so that $f^* = p_{\theta_0, \psi_0}$. Since $\alpha_0 < \alpha^*$, $g_{\theta_0}$ has heavier tails than $f^*$ and hence $\psi_0$ is bounded and continuous with $\psi_0(1) = 0$. Consequently, the density
$\psi_1(u) := (1 - \delta)\psi_0(u) + \delta, ~~ u \in (0,1),$
is bounded and continuous, and is bounded above $\delta$, and therefore, $\omega_1 = \log \psi_1$ can be extended to an element of $C[0,1]$. Now, for any $\psi = \scL(\omega)$ with $\omega \in C[0,1]$,
$d_{\kl}(f^*, p_{\theta_0, \psi}) = d_{\kl}(\psi_0, \psi) \le d_{\kl}(\psi_0, \psi_1) + \int_0^1 \psi_0(u) \log \frac{\psi_1(u)}{\psi(u)}du \le -\log(1-\delta) + 2\|\omega - \omega_1\|_\infty.$
Therefore, $\textstyle \Pr(d_{\kl}(f^*, p_{\theta, \psi}) < \epsilon \mid \theta = \theta_0) \ge \Pr(\|\omega - \omega_1\|_\infty < \frac\epsilon2)$,
where the latter probability, calculated for a Gaussian process $\omega$ with mean zero and covariance $\kappa^2 c_\lambda$, must be positive because such a Gaussian process has the entire $C[0,1]$ in its uniform topology support \citep{tokdar2007posterior,van2009adaptive}. An application of the law of total probability completes the proof.
\end{proof}

Although not apparent from the above result, the covariance parameters play an important role in determining how the prior mass is distributed within the broad support. {The inverse length-scale parameter $\lambda$ is of critical importance here because of its direct influence on the range of smoothing; though in our experience, a prior on $\kappa$ also helps with model fit and posterior computation via Markov chain Monte Carlo. We take $\pi_{\kappa^2}$ to be a convenient inverse-gamma distribution with shape $a_\kappa$ and rate $b_\kappa$, i.e., $(1/\kappa^2) \sim \gam(a_\kappa, b_\kappa)$ which is partially conjugate to the likelihood function in $\kappa^2$ and allows this parameter to be integrated out during model fitting. No such conjugate choice exists for $\lambda$ and formal  subjective or objective principles are difficult to apply in selecting $\pi_\lambda$; however, see \citet{paulo2005default,gu2018robust} for relevant discussions.}

{An alternative track is to seek $\pi_\lambda$ that guarantees optimal asymptotic frequentist convergence of the posterior distribution to the truth. In the setting of purely nonparametric density estimation with LGP, \cite{van2009adaptive} show that a gamma prior distribution on $\lambda$ is critical to optimally spreading prior mass into various smoothness classes, which in turn is critical to guaranteeing adaptive and optimal concentration of the posterior distribution to the truth. We follow this recommendation to specify $\pi_\lambda \sim \gam(a_\lambda, b_\lambda)$. Our numerical experiments were carried out with $a_\kappa = b_\kappa = 3/2$, $a_\lambda = 16$ and $b_\lambda=2.2$. The latter choices could be appreciated in several ways. Consider $\rho = c_\lambda(0,\Delta) = e^{-\lambda^2\Delta^2}$ which gives the correlation of the Gaussian process at a distance $\Delta$. With $\Delta = 10\%$, our choice of $\pi_\lambda$ assigns 95\% prior probability to $\rho \in (0.28, 0.84)$ with prior mean and median $\approx 0.6$. Alternatively, one could look at the number of up-crossings at zero of the process sample paths; which could be taken as a proxy to the number of local modes.  The well known Rice formula states that the expected number of up-crossings of zero of a mean zero Gaussian process on the unit interval with covariance $\kappa^2 c_\lambda$ is $\lambda/(\pi \sqrt{2}) \approx 0.22\lambda$ \citep{rice1944mathematical, adler2009random}. With our choice of $\pi_\lambda$, the prior probabilities of zero through five up-crossings, respectively, are 8\%, 37\%, 40\%, 13\% and 2\%.}


\subsection{Posterior computation}
For data $(y_1,\ldots,y_n)$, the likelihood function $(\theta,\psi) \mapsto \prod_{i = 1}^n\{g_\theta(y_i)\psi(u_i)\}$, where $u_i = G_\theta(y_i)$, involves $\psi$ only through the finite vector $\psi_U = (\psi(u_1), \ldots, \psi(u_n))^\top$. Unfortunately, the joint prior density of $\psi_U$, given model hyper-parameters $(\lambda,\kappa)$, is not available in closed form. This necessitates involving the latent Gaussian process $\omega$ in the representation $\psi = \scL(\omega)$ in posterior computation. However, $\psi_U$ depends on both the corresponding vector $\omega_U$ and the scalar $\omega_\norm = \int_0^1 e^{\omega(u)}du$ giving the normalization in the logistic transform and involving the whole function $\omega(u)$. It is practically impossible to carry out any numerical analysis of the posterior when a function valued input variable is involved in the likelihood evaluation. We overcome this challenge by adopting a grid-based representation of $\omega$ proposed and analyzed in \citet{tokdar2007towards}.

\subsubsection{Likelihood approximation}
Specifically, a dense set of points $T = \{0 = t_1 < t_2 < \cdots < t_L = 1\} \subset [0,1]$ is chosen as a grid over which both $\omega$ and $\psi$ are to be represented, respectively, as the vector $\omega_T = (\omega(t_1), \ldots, \omega(t_L))^\top$ and the corresponding vector $\psi_T$. Given $\omega_T$, a very accurate approximation to $\omega_\norm$ could be obtained by applying the trapezoidal rule of numerical integration to the pair $(T,\omega_T)$, readily producing the vector $\psi_T$. To evaluate $\psi_U$, which is needed for likelihood evaluation, it is useful to formally express the trapezoidal approximation to $\omega_\norm$ as the exact integration of the function $h(u)$ that linearly interpolates the points $(t_l, e^{\omega(t_l)})$, $l = 1, \ldots, L$. We may now view $\psi_T$ as the evaluation over the grid $T$ of the (normalized) density function $\bar h(u) = h(u) / \int_0^1h(t)dt$. Consequently, $\psi_U$ could be readily equated with the corresponding vector $\bar h_U$. The overall computational complexity of this likelihood approximation is $O(\max(n,L))$ and can be carried out extremely fast in actual time with optimized codes. In the numerical experiments reported here we use an equally spaced grid with $L = 101$ and increment size 0.01. 

\subsubsection{Low rank approximation and marginalization of hyper-parameters}
With the availability of a grid based representation and the linear interpolation based approximation to the likelihood function, it is feasible to carry out a Markov chain Monte Carlo approximation the posterior distribution of $(\theta, \omega_T)$. The prior density of $\omega_T$, given $(\lambda,\kappa)$, is a multivariate normal density with mean zero and covariance $\kappa^2C_T(\lambda)$ where $C_T(\lambda) = ((c_\lambda(t_l, t_k)))_{l,k = 1}^L$. An evaluation of this density involves factorizing $C_T(\lambda)$ at $O(L^3)$ computational complexity, which is practicable but slow at $L = 101$, and could be outright prohibitive for larger grid sizes. Additionally, running a Markov chain sampler on $\omega_T$, which is a dense representation of a smooth function, produces slow-mixing chains. 

Considerable efficiency gains can be made by replacing the smooth Gaussian process $\omega$ with a low-rank Gaussian process \citep{snelson2006sparse, tokdar2007towards, banerjee2008gaussian}. For a set of {\it knots} $S = \{s_1, \ldots, s_m\} \subset [0,1]$, with $m$ much smaller than $L$, the so called predictive process $\tilde \omega(u) = \expct[\omega(u) \mid \omega(s_1), \ldots, \omega(s_m)]$, gives a smooth interpolation of the graph of $(S, \omega_{S})$, and is fully determined by the random vector $\omega_{S} = (\omega(s_1), \ldots, \omega(s_m))^\top$. Typically the predictive process conditioning is defined for given covariance parameters $(\lambda, \kappa)$, but a hyper-parameter marginalized extension proposed in \cite{yang2017joint} and described below offers considerable additional speed up. 

Integrate out $\kappa^2$ from the model and express the prior distribution of $\omega_S$ given $\lambda$ as the multivariate Student-t distribution with pdf $p(\omega_S | \lambda) \propto \{1 + \omega_S^\top C_S(\lambda)^{-1}\omega_S / (2b_\kappa)\}^{-(a_\kappa + m/2)}$ where $C_S(\lambda)$ is the analogue of $C_T(\lambda)$ over the knot set $S$. It is impossible to analytically integrate out $\lambda$, but a discrete integration could be carried out by replacing $\pi_\lambda$ with a discrete approximation over a dense set of support points $\{\lambda_1, \ldots, \lambda_G\} \subset (0,\infty)$. With $\pi^*_\lambda(g):=\Pr(\lambda = \lambda_g)$, the prior distribution of $\omega_S$ is the mixture density $p(\omega_S) = \sum_g \pi^*_\lambda(g) p(\omega_S | \lambda_g)$. The vector $\tilde \omega_T$, which is the predictive process replacement of $\omega_T$, can be computed analytically from $\omega_S$ as $\tilde \omega_T = \sum_g \pi^*_\lambda(g|\omega_S) A_g \omega_S$ where $A_g = C_{TS}(\lambda_g)C_S(\lambda_g)^{-1}$ with $C_{TS}(\lambda)$ denoting the $L\times m$ matrix with elements $c_\lambda(t,s)$, $t \in T, s \in S$, and, $\pi^*_\lambda(g|\omega_S) \propto \pi^*_\lambda(g) p(\omega_S | \lambda_g)$. 
We select the support points $0 < \lambda_1 < \cdots < \lambda_G$ of $\pi^*_\lambda$ based on the knots set $S$. First $\lambda_1$ is fixed such that $\rho_1 := c_{\lambda_1}(0, 0.1) = 0.95$ and then successive $\lambda_g$ values are chosen so that $d_\kl(N(0,C_S(\lambda_{g-1})),N(0,C_S(\lambda_g))) = 0.5$ until we get $\rho_{G+1} := c_{\lambda_{G+1}}(0,0.1) < 0.2$. This gradual stepping down ensures successive $N(0, \kappa^2 C_S(\lambda_g))$ distributions maintain considerable overlap, eliminating any major gaps in the prior distribution of $\omega_S$ due to the discretization of $\lambda$. In our experience, posterior calculation is not sensitive to exact choices of the bookending values of $\rho_1$ and $\rho_{G+1}$, or the Kullback-Leibler stepping size.


\subsubsection{Markov chain sampling and runtimes}
With the above approximations in place, the model parameters reduce to the $(m+2)$ dimensional vector $(\alpha, \sigma, \omega_S)$. An adaptive, blocked Metropolis sampler is used on a transformed parameter space such that multivariate normal proposals can be used. Candidate proposal covariances are slowly adapted to achieve a 15\% acceptance rate using Algorithm 4 of \cite{andrieu2008tutorial}. Results presented in this paper were achieved by using one block containing $\omega_S$, one block updating $\theta = (\alpha,\sigma)$, and one block updating all $(m+2)$ parameters simultaneously. An important consequence of the discretization of $\lambda$ is that all relevant matrices, namely $\{(A_g, R_g): g = 1, \ldots, G\}$, where $R_g$ is a Cholesky factor of $C_S(\lambda_g)$, could be precomputed and stored prior to Markov chain sampling. Subsequent evaluations of the log posterior density reduce to $O(m \cdot \max(n,L))$ computing complexity. 

Experience suggests that the actual runtime of the sampler scales linearly in the sample size and sub-linearly in the number of knots $m$ or the grid size $L$. 
All numerical results reported in Sections \ref{sec:finite-samp} and \ref{sec:real-data} use $L=101, m = 11$, with equally spaced points  in $T$ and $S$ with end points equalling 0 and 1. This choice of $S$ leads to a discretization of $\pi_\lambda$ with $G = 30$ support points. In analyzing Fort Collins precipitation data with sample size $n = 6180$ (Section \ref{sec:real-data}), it took 9.8 minutes on a personal computer to carry out 500,000 iterations of the Markov chain. For two further subsamples with $n = 3645$ (0.6x with respect to the original set) and $n = 1061$ (0.2x), the same number of iterations took 6.3 (0.6x) and 2.2 (0.2x) minutes respectively. For the original set with $n = 6180$, it took 12.2 minutes (1.2x) to run the same number of iterations when the knots set $S$ was doubled to $m=21$ equally spaced knots ($G = 82$), keeping $L$ fixed at 101. Similarly, when the grid $T$ was doubled to $L = 201$ grid points (retaining $m = 11,G = 30$), it took 12.5 minutes (1.3x) to run the same number of iterations. We recommend $L = 101$ and $m = 11$ as default choices. But for any application, one should assess whether finer approximations are needed by repeating the analysis with larger values of $L$ and $m$ until posterior calculations stabilize. 

\section{Asymptotic properties}
\label{sec:asymp-prop}
{In recent years, mathematical analyses of large sample concentration properties of the posterior distribution have proven useful to the question of prior allocation in Bayesian analysis of infinite dimensional models; see \cite{ghosal2017fundamentals} for a comprehensive overview. Here we focus on posterior consistency and posterior contraction rate properties of the semiparametric LGP prior. Our treatment involves distinct model space topologies suitable for assessing either density estimation accuracy or tail index estimation accuracy; keeping in mind that proximity of two densities in $L^1$ topology may not guarantee similar tail index values. Posterior consistency, a frequentist evaluation of a prior intended for Bayesian applications, guarantees intersubjective knowledge generation through asymptotic merger of beliefs \citep{diaconis1986consistency}.}


\subsection{Density estimation consistency}
\label{sec:l1-consis}
Let $\Pi$ denote the induced prior measure on $f = p_{\theta,\psi}$, with $(\theta, \psi) \sim \pi_\theta\times \lgp(0, \kappa^2 c_\lambda)$ where we treat {$\kappa =1 $ as fixed, and work with a gamma prior on $\lambda$}. We allow $\pi_\theta$ to be arbitrary but assume it has a compact support $\Theta = [\alphalo, \alphahi]\times [\sigmalo, \sigmahi]$ for some $0 < \alphalo < \alphahi < \infty$, $0 < \sigmalo < \sigmahi < \infty$, with a strictly positive density in the interior of $\Theta$. {Compactness of $\Theta$ is assumed chiefly for technical reasons. An unbounded support adds layers of complication to critical function approximation results used below (e.g., Lemma \ref{lem:h} in Appendix A) with little gain in insight. One can enlarge $\Theta$ arbitrarily without virtually altering the posterior contraction rates.} 

We may view $\Pi$ as a probability measure on $\scF$, the subspace of density functions in $L^1[0,\infty)$. Given data $(y_1, \ldots, y_n)$, the posterior measure equals $\Pi(df \mid y_1,\ldots, y_n) \propto \{\prod_{i=1}^n f(y_i)\} \Pi(df)$. Below $\Holder^\beta[0,1]$ denotes the H\"older-$\beta$ space consisting of functions on $[0,1]$ that are $\lfloor\beta\rfloor$ times continuously differentiable with the $\lfloor\beta\rfloor$-th derivative being H\"older continuous of exponent $\beta - \lfloor\beta\rfloor$ where $\lfloor\beta\rfloor$ is the largest integer smaller than $\beta$. The minimax density estimation rate over H\"older-$\beta$ classes is $n^{-\frac\beta{2\beta + 1}}$ \citep{stone1982optimal}.


\begin{thm}
If $Y_1, \ldots, Y_n \iid f^*$ where $f^*$ is a bounded, continuous, regularly varying density on $(0,\infty)$ with tail index $\alpha^* \in (\alphalo, \alphahi)$, then 
$
\plim_{n\to \infty} \Pi(\{f:\|f - f^*\|_1 > \epsilon\} \mid Y_1, \ldots, Y_n) = 0$
for every $\epsilon > 0$. {Additionally, if $f = p_{\theta^*, \psi^*}$ with $\theta^*$ in the interior of $\Theta$ and $\phi^* = \log \psi^* \in \Holder^\beta[0,1]$ for some $\beta > 2$, then the fixed error margin $\epsilon$ may be replaced with the vanishing sequence $\epsilon_n = Bn^{-\frac{\beta}{2\beta+1}} (\log n)^{\frac{4\beta+1}{2\beta + 1}}$ for some large constant $B$.}
\label{thm:l1-consis}
\end{thm}

\begin{proof}
{To prove the first claim, we only need to establish} \citep[Theorem 2]{ghosal1999posterior}
\begin{description}
\item[C1.] $\Pi(\{f:d_\kl( f^*, f) < \epsilon\}) > 0~\mbox{for every}~\epsilon > 0$,
\item[C2.] for any $\epsilon > 0$, there exist constants $c, C > 0$ and sets $\scF_1, \scF_2, \ldots \subset \scF$, such that $\Pi(\scF_n^c) \le c e^{-Cn}$ and $\log N(\epsilon, \scF_n, d_H) \le n\epsilon^2$ for all large $n$,
\end{description}
where $N(\delta, \scF_n, d_H)$ denotes the covering number of $\scF_n$ by balls of radius $\delta$ in the Hellinger metric $d_H(p,q) = [\int \{\sqrt p(y) - \sqrt q(y)\}^2 dy]^{1/2}$. C1 follows readily from Theorem \ref{thm:kl}. {C2 follows from the following stronger condition necessary for the second part of the theorem.
\begin{description}
\item[C2*.] For every $0 < t < 1/2, s > 0$, there exist a constant $C > 0$ and sets $\scF_1, \scF_2, \ldots \subset \scF$, such that $\Pi(\scF_n^c) \le e^{-(C+4)n\bar\epsilon_n^2}$ and $\log N(\bar\epsilon_n, \scF_n, d_H) \le n\epsilon_n^2$ for all large $n$, where $\bar\epsilon_n = Bn^{-t}(\log n)^s$ for some $B > 0$ and $\epsilon_n = \bar \epsilon_n \log n$.
\end{description}
%
A proof is given in Appendix C. A key step is Lemma \ref{lem:h} (Appendix A) which states $d_H(p_{\theta_1,\psi_1}, p_{\theta_2,\psi_2}) \le c_2\|\omega_1\|^{1/2}_{C^2} \|\theta_1 - \theta_2\| + \|\omega_1 - \omega_2\|_\infty e^{\|\omega_1 - \omega_2\|_\infty/2}$ if $\theta_1,\theta_2$ are interior points in $\Theta$ and $\psi_1=\scL(\omega_1)$, $\psi_2 = \scL(\omega_2)$ with $\omega_1, \omega_2 \in C^2[0,1]$, the space of twice continuously differentiable functions on $[0,1]$ with norm $\|\omega\|_{C^2} := \|\omega\|_\infty + \|\dot \omega\|_\infty + \|\ddot \omega\|_\infty$. Clearly, $\Holder^\beta[0,1] \subset C^2[0,1]$ for all $\beta >2$. Construction of the sets $\scF_1,\scF_2,\ldots$ relies on the observation that a separable, mean-zero Gaussian process $\omega$ on $[0,1]$ with covariance function $c_\lambda$ may be viewed as a Borel measurable random element with a Gaussian measure $\nu^\lambda$ on the Banach space $(C^2[0,1], \|\cdot\|_{C^2})$.  Our construction builds upon that of \cite{van2009adaptive} who embed the Gaussian measure in $(C[0,1], \|\cdot\|_\infty)$. However, some key modifications are needed to address the change in the embedding space (Appendix C). 

By Theorem 8.9 of \cite{ghosal2017fundamentals}, under the additional assumption on $f^*$, a proof of the second part of the theorem may be established by applying C2* with $t = \frac{\beta}{2\beta+1}$, $s=2t$, 
in conjunction with the following sharper version of C1:
\begin{description}
\item[C1*.] $\Pi(\{f:d_\kl( f^*, f) \le \bar\epsilon_n^2, V(f^*, f) \le \bar\epsilon_n^2\}) \ge e^{-Cn\bar\epsilon_n^2}$ for all large $n$,
\end{description}
where $V(f, g) = \int f(y)\log^2\{{f(y)}/{g(y)}\}dy$. This sharper prior concentration bound can be proved via a non-trivial extension of Theorem 3.1 of \citet{van2009adaptive}. A proof of possible independent interest is given in Appendix B. 
}
\end{proof}

\subsection{Tail index estimation consistency}
{

In the following, assume without loss of generality that $\alphalo \le \frac12$ and $\alphahi > 1$. As in Theorem \ref{thm:l1-consis}, assume that the true density is some $f^* = p_{\theta^*,\psi^*}$ where $\theta^* = (\alpha^*, \sigma^*)$ is in the interior of $\Theta$ and $\phi^* = \log \psi^* \in \Holder^\beta[0,1]$ with $\beta > 2$. 
Denote $\gamma = \frac{\beta}{2\beta + 1} \in (0,1/2)$ so that the posterior contraction rate in $L^1$ topology equals a constant multiple of $n^{-\gamma}(\log n)^{2\gamma+1}$. 
The lower bound assumption on $\beta$ implies that both $\bar\rho(\xi) := \frac{2\xi}{2\xi+1}\gamma - \frac{3(1 - 2\gamma)}{2\alpha^*(2\xi+1)}$ and $\hat \rho(\xi) := \xi \gamma - \frac32(1 - 2\gamma)$ are strictly positive for every $\xi \in [\frac{\alphalo}{\alpha^*}, 1]$. 
\begin{thm}
If $Y_1, \ldots, Y_n \iid f^*$ and ${\alphalo}/{\alpha^*} < \xi < \min(1,1/\alpha^*)$ {is such that $\beta\xi > 3/2$} then $\plim_{n \to \infty} \Pi(\{f: |\ti(f) - \alpha^*| >  B_1n^{-\rho}(\log n)^s\} \mid Y_1,\ldots,  Y_n) = 0$ for all large $B_1$ where $\rho = \min\{\bar \rho(\xi), \hat \rho(\xi)\}$ and $s = 2\rho + \frac{4}{\alpha^*(2\xi + 1)}$ if $\rho = \bar \rho(\xi)$, $s = 2\rho + 4$ otherwise.
\label{thm:tail-consis}
\end{thm}

A proof is presented in Appendix E. The main argument relies on establishing existence of tests that can distinguish $f^*$ from model elements $f=p_{(\alpha,\sigma),\psi}$ with $|\alpha^* - \alpha| > B_1n^{-\rho}(\log n)^s $ with type I and II error probabilities vanishing suitably rapidly. This line of argument directly follows the path laid out in the original work of \cite{schwartz1965bayes}; a modern presentation is Theorem 8.9 \cite{ghosal2017fundamentals}. See also \cite{kleijn2021frequentist} for related recent developments. \cite{li2019posterior} present a similar theoretical exploration with test functions derived from an exceedance probability based tail index estimator of \cite{carpentier2015adaptive}. Our proof relies on a more complex test procedure which first tries to detect a difference between the exceedance probability of the empirical distribution at a high threshold and that of the true distribution, and if no significant difference is detected then repeats the process one more time at an even higher threshold but only to the conditional distributions to the right of the first threshold.

The theorem requires sufficient smoothness of the true density via the assumption {$\beta \cdot \min(1, 1/{\alpha^*}) > 3/2$ so that a suitable $\xi$ may be found with $\beta\xi > 3/2$}. This condition demands that a relatively rough density {(small $\beta$)} must have a sufficiently heavy tail {(small $\alpha^*$)} to insure accurate estimation of the latter with our semiparametric estimation model. This requirement may be understood in the light that with a density function lacking in smoothness, the bulk of the density carries less information about the tail, and hence an accurate estimation of the tail is possible only when more observations are available directly from the tail itself, i.e., only when the tail is heavy.

Additionally, multiple factors control the value of $\rho$ which determines the posterior contraction rate. Notice that both $\bar \rho(\xi)$ and $\hat \rho(\xi)$ are strictly increasing in $\xi$ and hence sharpest rates are obtained by taking $\xi$ as close as possible to the maximum allowed value of $\min(1, \frac1{\alpha^*})$. Since $\bar \rho(\xi) < \hat \rho(\xi)$ if and only if $\beta\xi(2\xi - 1) > \frac32(2\xi + 1 - 1/\alpha^*)$, the following observations can be made on the fastest possible rate. If $\alpha^* \in (\alphalo, \frac{2 - \alphalo}{1 + \alphalo}]$ then Theorem \ref{thm:tail-consis} holds with any $\xi$ arbitrarily close to $\min(1,\frac1{\alpha^*})$, $\rho = \bar\rho(\xi)$  and $s = 2\rho + \frac{4}{\alpha^*(2\xi + 1)}$. On the other hand, if $\alpha^* \ge 2$, the theorem holds with any $\xi$ arbitrarily close to $\frac1{\alpha^*}$, $\rho = \hat \rho(\xi)$ and $s = 2\rho + 4$. In the intermediate case of $\alpha^* \in (\frac{2 - \alphalo}{1 + \alphalo}, 2)$, $\xi$ can be arbitrarily close to $\frac1{\alpha^*}$ with $\rho = \bar\rho(\xi)$, $s = 2\rho + \frac{4}{\alpha^*(2\xi + 1)}$ if $\beta > \frac{3\alpha^*}{2}\times\frac{1 + \alpha^*}{2 - \alpha^*}$ and $\rho =  \hat\rho(\xi)$ and $s = 2\rho + 4$ otherwise. 

When $\alpha^* \in (\alphalo, \frac{2 - \alphalo}{1 + \alphalo}]$, the choice of $\rho = \bar\rho(\xi) = \frac{2\xi}{2\xi + 1} \gamma - \frac{3(1- 2\gamma)}{2\alpha^*(2\xi + 1)}$ compares favorably to the optimal rates obtained by \cite{hall1984best, hall1985adaptive}. In particular, whenever $\psi^*$ is infinitely smooth, e.g., $f^*$ is a GPD itself, the density estimation contraction rate has $\gamma \approx \frac12$ and hence $\rho \approx \frac{\xi}{2\xi + 1}$ with $\xi \approx \min(1, \frac1{\alpha^*})$; here $\approx$ indicates ``arbitrarily close from below''. Since $|y^{\alpha^*}\bar F^*(y)/\si(f^*) - 1| \asymp y^{-\min(1,\frac1{\alpha^*})}$, with $\si(f) = \lim_{y \to \infty} y^{\alpha_+(f)}\bar F(y)$, 
$f^*$ belongs to a suitable Hall-Welsh class of heavy tailed densities $\scD(\alpha^*, C_0, \epsilon, \xi, A):= \{f:\bar F(y) = Cy^{-\alpha}\{1 + R(y)\}, |R(y)| < Ay^{-\xi\alpha}, |\alpha - \alpha^*| < \epsilon, |C - C_0| < \epsilon\}$ for which the minimax rate of tail index estimation is precisely $n^{-\frac{\xi}{2\xi + 1}}$. See Section \ref{sec:discussion} for further discussion.

}

\begin{table*}[!t]
\spacingset{1.0} 
\centering
\footnotesize
\begin{tabular}{r|r|r|rrr|llll}
\hline
Model  			&EVI		&Method	&\multicolumn{3}{c|}{Estimating $\xi$} & \multicolumn{4}{c}{Estimating $\bar Q(p)$ ($\mbox{rMAE}_{\mbox{\tiny Cover}}$)}\\
            			&		&		&Bias	&RMSE 	&Cover	&$p=0.01$ 				&$0.001$ 				&$10^{-4}$ 				&$10^{-5}$\\
\hline
\mrow{10}{GPD}	&\mrow{2}{0.1}	&Semi	&0.03  	&0.05   	&99	&$0.06_{\mbox{\tiny 93}}$	&$0.11_{\mbox{\tiny 94}}$ 	&$0.18_{\mbox{\tiny 96}}$ 	&$0.26_{\mbox{\tiny 98}}$\\
				&  			&Thresh	&0.14   	&0.16   	&84	&$0.06_{\mbox{\tiny 95}}$	&$0.17_{\mbox{\tiny 93}}$	&$0.44_{\mbox{\tiny 91}}$	&$0.90_{\mbox{\tiny 89}}$\\
				\cline{2-10}
		       		&\mrow{2}{0.2} &Semi	&0.01 	&0.06   	&100	&$0.07_{\mbox{\tiny 97}}$ 	&$0.14_{\mbox{\tiny 99}}$ 	&$0.23_{\mbox{\tiny 100}}$ 	&$0.34_{\mbox{\tiny 100}}$\\
				&			&Thresh	&0.12   	&0.15   	&90 	&$0.07_{\mbox{\tiny 96}}$	&$0.20_{\mbox{\tiny 97}}$	&$0.50_{\mbox{\tiny 94}}$	&$1.03_{\mbox{\tiny 92}}$\\
				\cline{2-10}
       				&\mrow{2}{0.3} 	&Semi	&0.02  	&0.06   	&99	&$0.08_{\mbox{\tiny 94}}$ 	&$0.18_{\mbox{\tiny 96}}$ 	&$0.30_{\mbox{\tiny 96}}$ 	&$0.46_{\mbox{\tiny 97}}$ \\
				&			&Thresh	&0.12   	&0.17   	&85	&$0.08_{\mbox{\tiny 96}}$	&$0.29_{\mbox{\tiny 93}}$	&$0.75_{\mbox{\tiny 88}}$	&$1.59_{\mbox{\tiny 86}}$\\
				\cline{2-10}
	       			&\mrow{2}{0.5} 	&Semi	&0.00 	&0.09   	&97	&$0.13_{\mbox{\tiny 93}}$ 	&$0.29_{\mbox{\tiny 93}}$ 	&$0.49_{\mbox{\tiny 95}}$ 	&$0.76_{\mbox{\tiny 96}}$ \\
				&			&Thresh	&0.09   	&0.15   	&89	&$0.14_{\mbox{\tiny 91}}$	&$0.41_{\mbox{\tiny 91}}$	&$0.92_{\mbox{\tiny 90}}$	&$1.81_{\mbox{\tiny 91}}$\\
				\cline{2-10}
	       			&\mrow{2}{1}	&Semi	&-0.01  	&0.12    	&97	&$0.23_{\mbox{\tiny 95}}$ 	&$0.49_{\mbox{\tiny 97}}$ 	&$0.85_{\mbox{\tiny 97}}$ 	&$1.45_{\mbox{\tiny 97}}$\\
				&			&Thresh	&0.04   	&0.15   	&89	&$0.25_{\mbox{\tiny 94}}$	&$0.62_{\mbox{\tiny 92}}$	&$1.23_{\mbox{\tiny 91}}$	&$2.33_{\mbox{\tiny 91}}$\\
				\hline
\mrow{10}{GPD4}	&\mrow{2}{0.1} 	&Semi	&0.03  	&0.09   	&97	& $0.05_{\mbox{\tiny 95}}$ 	&$0.12_{\mbox{\tiny 91}}$ 	&$0.25_{\mbox{\tiny 90}}$ 	&$0.48_{\mbox{\tiny 91}}$ \\
				&			&Thresh	&0.13   	&0.15   	&83	&$0.04_{\mbox{\tiny 100}}$	&$0.15_{\mbox{\tiny 95}}$	&$0.40_{\mbox{\tiny 93}}$	&$0.80_{\mbox{\tiny 88}}$\\
				\cline{2-10}
				&\mrow{2}{0.2} 	&Semi	&0.01 	&0.06   	&98	&$0.06_{\mbox{\tiny 95}}$ 	&$0.13_{\mbox{\tiny 95}}$ 	&$0.22_{\mbox{\tiny 96}}$ 	&$0.34_{\mbox{\tiny 98}}$\\  			
				&			&Thresh	&0.08 	&0.11   	&92	&$0.06_{\mbox{\tiny 96}}$	&$0.17_{\mbox{\tiny 97}}$	&$0.38_{\mbox{\tiny 95}}$	&$0.71_{\mbox{\tiny 91}}$\\
				\cline{2-10}
				&\mrow{2}{0.3} 	&Semi	&-0.00 	&0.07   	&96	&$0.08_{\mbox{\tiny 97}}$ 	&$0.17_{\mbox{\tiny 94}}$ 	&$0.29_{\mbox{\tiny 95}}$ 	&$0.43_{\mbox{\tiny 96}}$\\  			
				&			&Thresh	&0.07   	&0.12   	&93	&$0.08_{\mbox{\tiny 95}}$	&$0.23_{\mbox{\tiny 97}}$	&$0.51_{\mbox{\tiny 94}}$	&$0.97_{\mbox{\tiny 91}}$\\
				\cline{2-10}
				&\mrow{2}{0.5} &Semi	&0.02  	&0.08   	&92	&$0.13_{\mbox{\tiny 94}}$ 	&$0.30_{\mbox{\tiny 93}}$ 	&$0.53_{\mbox{\tiny 91}}$ 	&$0.83_{\mbox{\tiny 92}}$\\  		
				&			&Thresh	&0.06   	&0.13   	&93	&$0.15_{\mbox{\tiny 93}}$	&$0.39_{\mbox{\tiny 95}}$	&$0.79_{\mbox{\tiny 93}}$	&$1.42_{\mbox{\tiny 92}}$\\
				\cline{2-10}
				&\mrow{2}{1}    	&Semi	&0.05  	&0.10   	&96	&$0.17_{\mbox{\tiny 97}}$ 	&$0.40_{\mbox{\tiny 96}}$ 	&$0.76_{\mbox{\tiny 96}}$ 	&$1.36_{\mbox{\tiny 96}}$\\  		
				&			&Thresh	&0.04   	&0.12   	&97	&$0.19_{\mbox{\tiny 94}}$	&$0.44_{\mbox{\tiny 96}}$	&$0.82_{\mbox{\tiny 96}}$	&$1.38_{\mbox{\tiny 96}}$\\
\hline
\mrow{10}{Half-t}	&\mrow{2}{0.1} 	&Semi	&-0.06  	&0.06   	&87	&$0.06_{\mbox{\tiny 79}}$ 	&$0.12_{\mbox{\tiny 72}}$ 	&$0.16_{\mbox{\tiny 83}}$ 	&$0.18_{\mbox{\tiny 91}}$\\ 			
				&			&Thresh	&0.09   	&0.11   	&96	&$0.04_{\mbox{\tiny 92}}$	&$0.13_{\mbox{\tiny 90}}$	&$0.31_{\mbox{\tiny 86}}$	&$0.59_{\mbox{\tiny 87}}$\\
				\cline{2-10}
				&\mrow{2}{0.2} 	&Semi	&-0.10   	&0.11    	&56	&$0.06_{\mbox{\tiny 92}}$ 	&$0.11_{\mbox{\tiny 94}}$ 	&$0.17_{\mbox{\tiny 97}}$ 	&$0.27_{\mbox{\tiny 95}}$ \\ 			
				&			&Thresh	&0.06   	&0.11   	&94	&$0.06_{\mbox{\tiny 91}}$	&$0.16_{\mbox{\tiny 94}}$	&$0.37_{\mbox{\tiny 95}}$	&$0.71_{\mbox{\tiny 96}}$\\
				\cline{2-10}
				&\mrow{2}{0.3}	&Semi	&-0.12   	&0.13    	&75	&$0.07_{\mbox{\tiny 96}}$ 	&$0.13_{\mbox{\tiny 97}}$ 	&$0.23_{\mbox{\tiny 96}}$ 	&$0.35_{\mbox{\tiny 90}}$ \\ 			
				&			&Thresh	&0.04  	&0.11   	&99	&$0.06_{\mbox{\tiny 98}}$	&$0.20_{\mbox{\tiny 96}}$	&$0.44_{\mbox{\tiny 96}}$	&$0.78_{\mbox{\tiny 97}}$\\
				\cline{2-10}
				&\mrow{2}{0.5} 	&Semi	&-0.10   	&0.13    	&87	&$0.10_{\mbox{\tiny 95}}$ 	&$0.23_{\mbox{\tiny 96}}$ 	&$0.39_{\mbox{\tiny 96}}$ 	&$0.54_{\mbox{\tiny 91}}$ \\ 			
				&			&Thresh	&0.02 	&0.10   	&95	&$0.10_{\mbox{\tiny 96}}$	&$0.26_{\mbox{\tiny 96}}$	&$0.51_{\mbox{\tiny 96}}$	&$0.85_{\mbox{\tiny 95}}$\\
				\cline{2-10}
				&\mrow{2}{1}   	&Semi	&-0.06  	&0.14    	&96	&$0.19_{\mbox{\tiny 96}}$ 	&$0.36_{\mbox{\tiny 96}}$ 	&$0.57_{\mbox{\tiny 96}}$ 	&$0.82_{\mbox{\tiny 96}}$\\  			
				&			&Thresh	&0.02   	&0.11   	&100&$0.21_{\mbox{\tiny 95}}$	&$0.45_{\mbox{\tiny 98}}$	&$0.76_{\mbox{\tiny 99}}$	&$1.19_{\mbox{\tiny 99}}$\\
				\hline
\end{tabular}
\caption{Estimating the extreme value index (EVI) $\xi = \alpha^{-1}$ and high tail quantiles $\bar Q(p) = \bar F^{-1}(p)$ from synthetic data of sample size $n = 1000$. Estimation accuracy and coverage of 95\% credible intervals are averaged across 100 data sets for each experimental group. For $\bar Q(p)$, estimation accuracy is measured via relative mean absolute error as a fraction of the true quantile value. Additional keys: RMSE = root mean squared error, Cover = coverage.}
\label{tab:gpds-long}
\end{table*}

\section{Finite sample behavior}
\label{sec:finite-samp}
\subsection{Tail index estimation}
\label{sec:tix}
{From \cite{hall1984best}, statistical performance of any estimator of tail quantities depends on how quickly the actual tail starts resembling the corresponding Pareto tail $y^{-\ti(f)}$. 
When sample size $n$ is only moderately large, the Pareto shape may only be partially established within the range of the observed data, posing a serious challenge to any thresholding method in detecting if and where a bulk-to-tail transition takes place. A similar challenge is posed to our joint semiparametric estimation which must balance a likelihood function that receives little information from a partially established tail against a model specification that idealizes a generalized Pareto-like tail.}

Consider three different choices of the shape of $f$, namely, (i) GPD: $f_\alpha(y) = g_{(\alpha,1)}(y)$, (ii) GPD4: $f_\alpha(y) = 4g_{(\alpha, 1)}(y)\{G_{(\alpha,1)}(y)\}^3$, and (iii) Half-t: $f_\alpha(y) = 2c(\alpha)(1 + y^2/\alpha)^{-(\alpha + 2)/2}$, $c(\alpha) = \Gamma(\frac{\alpha+1}2)/\{\sqrt{\alpha \pi} \Gamma(\frac\alpha2)\}$; {each giving a regularly varying density with tail index} $\alpha$. Table \ref{tab:gpds-long} reports performance statistics of our semiparametric estimation of the corresponding extreme value index $\xi = \alpha^{-1}$, averaged across 100 data sets of size $n = 1000$ each, with the true value of $\xi$ varying over $\{0.1, 0.2, 0.3, 0.5, 1.0\}$. For comparison, we include corresponding figures from a thresholding estimation of $\xi$, where the threshold is determined by the adaptive technique of \cite{durrieu2015nonparametric}, followed by a Bayesian fit of a GPD model to the excess data with the GPD location parameter set at the threshold, and the scale $\sigma$ and shape $1/\alpha$ estimated under the same prior as used in our semiparametric estimation. 

For the GPD sets, in addition to smaller bias and averaged error for the point estimates, the 95\% posterior credible intervals from the semiparametric method are much narrower with higher coverage than those from the thresholding method (figure included in supplementary material). This improvement is unsurprising; the true density $f_\alpha$ matches the {model} specification in a very strong way. A similar match between the {model} and the truth is absent in the GPD4 sets for which $f_\alpha$ may be expressed as $p_{(\alpha,\sigma), \psi}$ but only with a $\psi = \scL(\omega)$ for which $\lim_{u \to 0} \omega(u) = -\infty$. But this misspecification at the left does not appear to affect estimation of the right tail, where the semiparametric method performs as well or better than the thresholding method, especially when the tail is not too heavy. 

The Half-t sets pose a far more serious challenge to the semiparametric method. Although the averaged error of the estimates are comparable between the two methods, the semiparametric method incurs a strong negative bias for $\xi = \alpha^{-1}$ (i.e., underestimates the tail heaviness) with fairly tight posterior credible intervals resulting in poor coverage when true $\xi < 1$. For a half-t density, we may use \eqref{eq:psi} to write $f_\alpha = p_{(\alpha, \sigma),\psi}$ and verify that $\phi = \log \psi \in C[0,1]$ but 
$\dot \phi(1 - t) = \frac{(\alpha + 1)t^{\xi-1}}{\alpha} \times \frac{(1 + \alpha\sigma^2)t^\xi - \alpha\sigma^2}{t^{2\xi} + \alpha\sigma^2 (1 - t^\xi)^{2}},$
and hence $\lim_{u \to 1} \dot \phi(u)$ equals $-\infty, -2$ or 0, according to whether $\xi < 1$, $\xi=1$ or $\xi>1$. There is a strong mismatch between the {idealized shape} and the truth on the right tail when $\xi < 1$, causing the posterior distribution on $\xi$ to be biased downward.

\subsection{Estimation of tail quantiles} Although an accurate estimation of the tail index parameter is conceptually appealing, practical interest usually focuses on estimating tail quantiles of $f$. By extending the numerical analysis presented above, we find that the semiparametric joint estimation is substantially more effective at this task than the thresholding approach. 
Specifically, we look at the estimates and the 95\% posterior credible intervals of the tail quantiles $\bar Q(p) = \bar F^{-1}(p) = F^{-1}(1-p)$ associated with excess tail probabilities $p \in \{0.01, 0.001, 0.0001, 0.00001\}$ and compare these against the true values for the $3\times 5$ experimental sets reported above. Both methods produce credible intervals with coverage at or above the nominal 95\% level in most cases, but the semiparametric estimate is typically more accurate than the thresholding estimate, with up to 400\% improvement in some cases for very high quantiles (Table \ref{tab:gpds-long}). The semiparametric posterior credible interval is also much tighter than the threshold based interval (not shown). 

The only concern about coverage of the semiparametric credible interval arises in the Half-t sets with a  small $\xi$, for which the semiparametric model is strongly misspecified at the right tail. However, a closer inspection of these cases reveals that while the semiparametric method overestimates the high quantiles, it still gives a credible interval that is comparable in magnitude to the true quantile value. In contrast, the thresholding method may minimally contain the true value at the lower end of its interval but usually produces a very wide interval with the upper end of the interval being several orders of magnitudes larger than the truth (Figure \ref{fig:qtail}). In other words, in spite of the persistent bias in estimating asymptotic tail heaviness, the semiparametric method produces reasonably accurate and meaningful estimates of the tail itself.



\begin{figure*}[!t]
\centering
\includegraphics[width=\textwidth]{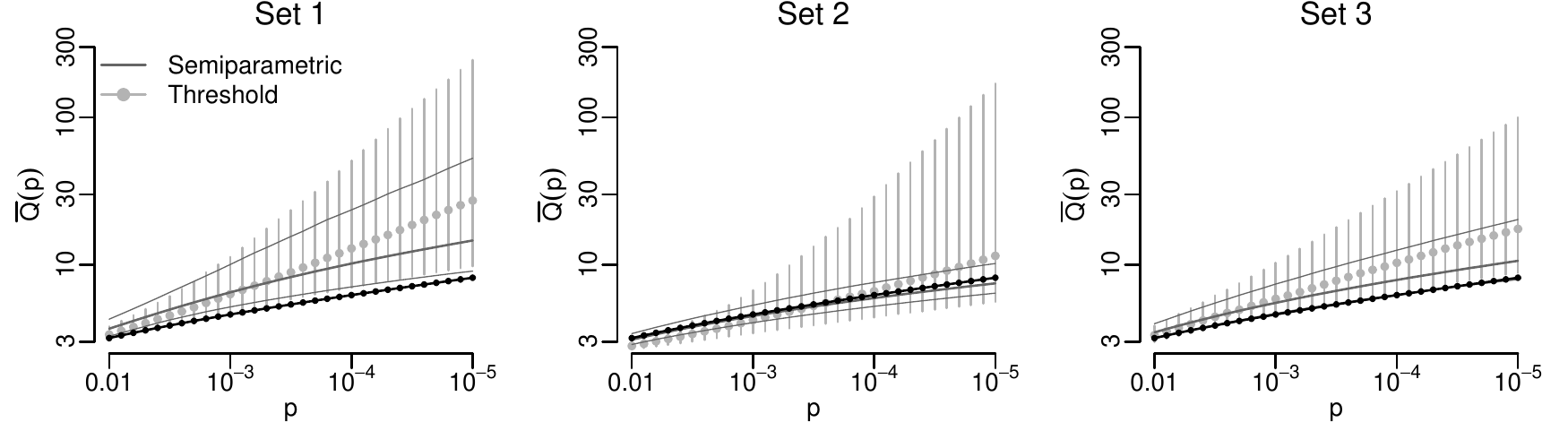}
\caption{A comparison of the 95\% posterior credible intervals for $\bar Q(p) = \bar F^{-1}(p)$ from the semiparametric and the thresholding methods, for three randomly chosen Half-t sets with $\xi = \alpha^{-1} = 0.1$, for which the semiparametric method serisouly underestimates $\xi$. True quantile values are shown as connected black beads. }
\label{fig:qtail}
\end{figure*}

\section{Fort Collins precipitation}
\label{sec:real-data}
\cite{katz2002statistics} present an analysis of total daily precipitation measurements (in inches) between 1900-1999 from a single rain gauge in Fort Collins, CO, estimating a heavy-tailed distribution with $\xi = \alpha^{-1} = 0.18$ at the threshold of 0.4 inches. \citet{scarrot2012review} estimate $\hat\xi = 0.21 \pm 0.04$ (standard error) at a similar threshold, and identify two additional candidates for the threshold value at which usual GPD diagnostics plots appear to stabilize, each leading to a different estimate of the tail index parameter: $\hat \xi = 0.13 \pm 0.07$ at threshold 0.85 and $\hat \xi = 0.003 \pm 0.09$ at threshold 1.2. This kind of ambiguity about the tail index is distinct from pure statistical uncertainty resulting from sampling variability. A Bayesian expression of joint uncertainty of the bulk and the tail could be particularly useful in mitigating between multiple distinct GPD tails offering partial match. 

The original data set\footnote{Taken from the \texttt{extRemes} package in \texttt{R} \citep{gilleland2011new}.} contains $N=36,524$ daily measurements with 78\% of the records being zero; the rest are recorded to the nearest hundredth of an inch. We remove all records with a precipitation measurement below 0.03 inches and jitter the remaining data ($n = 6180$, 17\% of all records) with a small uniform noise between $-0.005$ and $0.005$ to break ties while preserving original precision. With a smooth LGP prior at the core, the semiparametric method is sensitive to the presence of strong discontinuous features in the data histogram. A large number of ties in the records is one such feature, which necessitates the random jittering. The presence of excess zeros is another such feature, which cannot be overcome by  jittering alone, since the distribution of the jittered data still presents a big jump discontinuity near zero. In fact, we find that such an effect persists up to measurements of 0.02 inches, whose inclusion in the data analysis significantly distorts the posterior inference from what is obtained when analyzing all or some subset of records $\ge 0.03$ inches. We return to this point below after presenting our results.

Figure \ref{fig:fcl} (left panel) shows thresholding estimates of $\xi = \alpha^{-1}$ obtained from a Bayesian fit of a GPD tail to excess data over the threshold, as described in Section \ref{sec:tix}. These estimates of $\xi$ are different from those reported in \citet{scarrot2012review}, who employ maximum likelihood estimation without restricting $\xi > 0$ and without any regularization via a prior. However, the detailed analysis of \cite{katz2002statistics} offers strong evidence of a heavy tail (i.e., $\xi > 0$), and thus a Bayesian estimation with a relatively flat prior on $\xi \in [0,2]$ appears a better alternative. In spite of a weak prior specification, the posterior estimate and interval of $\xi$ are heavily influenced by the prior choice for large threshold values at which little excess data is left for parameter estimation. The adaptive threshold choice method of \cite{durrieu2015nonparametric} gives a threshold value of 0.93, for which $\xi$ is estimated to be 0.22 with a 95\% posterior credible interval $[0.08, 0.41]$. 

\begin{figure*}[!t]
\centering
\includegraphics[width=0.32\textwidth]{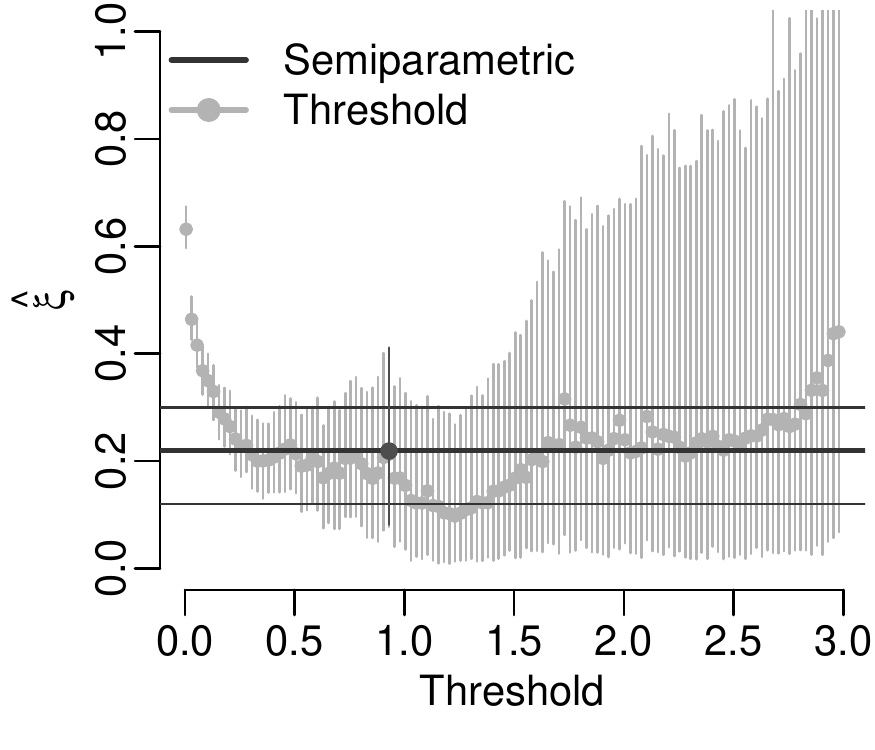}
\includegraphics[width=0.32\textwidth]{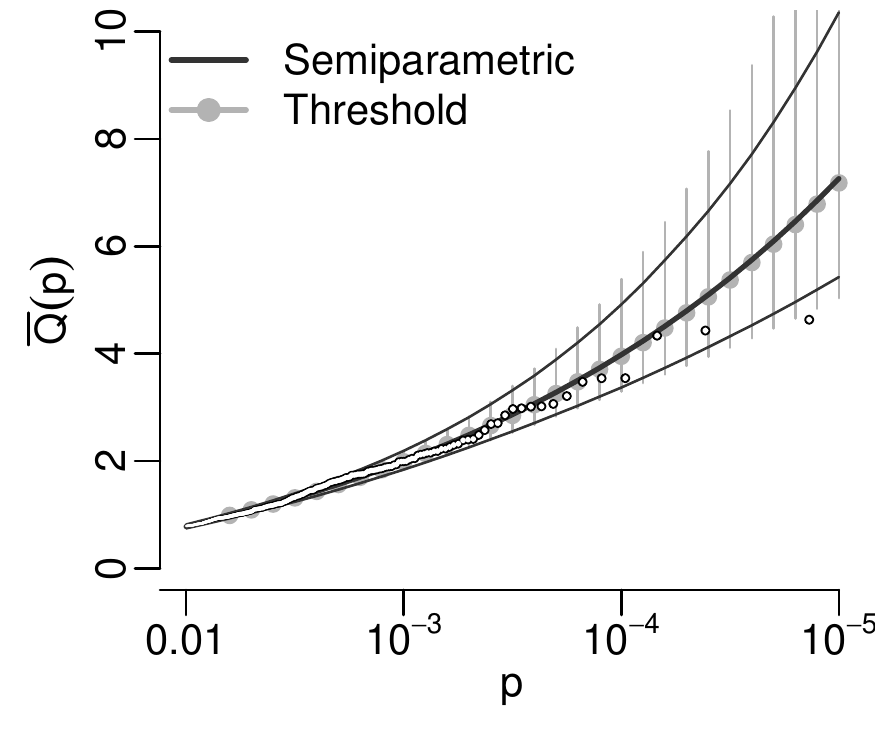}
\caption{Estimation of tail heaviness ($\xi = \alpha^{-1}$) and high tail quantiles for Fort Collins daily precipitation. Left panel shows thresholding estimates and 95\% credible intervals of $\xi$ corresponding to a grid of threshold values between 0.005 and 3.0 with an increment of 0.025; the adaptive choice of threshold = 0.93 is highlighted with a darker shade. The horizontal lines give the estimate and the 95\% credible interval for the semiparametric analysis. Right panel shows estimates of high quantiles $\bar Q(p) = \bar F^{-1}(p)$ from the semiparametric method and the thresholding method (threshold = 0.93). The exceedance probability $p$ corresponds to the original data of size $N = 36,524$, without any truncation or thresholding. A graph of the points $\{(\frac{i - 0.5}{N}, Y_{(i)}), 1 \le i \le N\}$ is included to visualize empirical quantiles, where $Y_{(i)}$ denotes the $i$-th order statistic of the original data.}
\label{fig:fcl}
\end{figure*}

The semiparametric method offers a comparable estimate of $\xi = 0.22$ with a tighter 95\% credible interval [0.12, 0.30]. Both methods point to a slightly heavier tail than what was reported by \cite{katz2002statistics}, but their estimate of $\xi = 0.18$ lies well within the 95\% credible intervals. The estimated high tail quantiles from the semiparametric method and the threshold method (threshold = 0.93) are very similar to one another and they line up well against empirical quantiles, but the 95\% credible intervals from the semiparametric method are considerably tighter (Figure \ref{fig:fcl}, right). However, the difference is much less stark than what we see in simulation studies.

The maximum daily precipitation during the observation period was 4.63 inches, recorded in the year 1997. The semiparametric method estimates the corresponding return period to be 47.6 years, with a 95\% posterior credible interval (PCI) of [23, 122.3]; the thresholding method gives similar estimates. These estimates are close to the estimated return period of 50.8 years reported by \cite{katz2002statistics}, who did not report an interval. The estimated return periods for 3 inches and 4 inches of precipitation are, respectively, 10 years (95\% PCI = [6.5, 16.5]) and 28 years (95\% PCI = [14.9, 59.9]). We note that in the 100 year observation period, there were 10 instances with 3 inches or more daily precipitation (1902, '04, '38, '49, '51, '51, '61, '77, '90, '97), of which three had more than 4 inches of rain ('02, '77, '97). More speculatively, we estimate the return period of 5 inches of rain to be 64.2 years (95\% PCI = [28.7, 178.8]). 

The estimates from the semiparametric method remain reasonably robust when analyzing further subsets of the data. When data analysis is restricted to records $\ge 0.1$ inches (or $\ge 0.4$ inches), the estimate of $\xi$ is $0.19$ with 95\% PCI = [0.06, 0.30] (or 0.16 with 95\% PCI = [0.04, 0.32]). For these further truncations, the estimated tail heaviness is slightly lower with greater uncertainty, but the upper end of the credible interval remains essentially the same. The same is reflected in high tail quantile estimates (Figure \ref{fig:fc-sens}). It appears that there is no strong evidence in the data pointing to a substantially lower tail heaviness than what was presented in \cite{katz2002statistics}. The possible lower estimates at higher threshold values discussed by \cite{scarrot2012review} are likely spurious. 

However, the semiparametric method is not completely robust to the issue of truncation. When data analysis is expanded to include all non-zero records, the posterior shifts substantially and results in a heavier tail estimate (0.3 with 95\% PCI = [0.25, 0.36]) with high tail quantiles being significantly larger than the estimates reported above (not shown). The same shift is noticed also when expanding the analysis only slightly to include records of 0.02 inches, or records of 0.01 and 0.02 inches. As indicated earlier, this discrepancy is likely an artifact of excess of zero and other tiny measurements which cannot be fully mitigated by jittering alone. See Section \ref{sec:discussion} for further discussion.

\begin{figure*}[!t]
\centering
\includegraphics[width=0.32\textwidth]{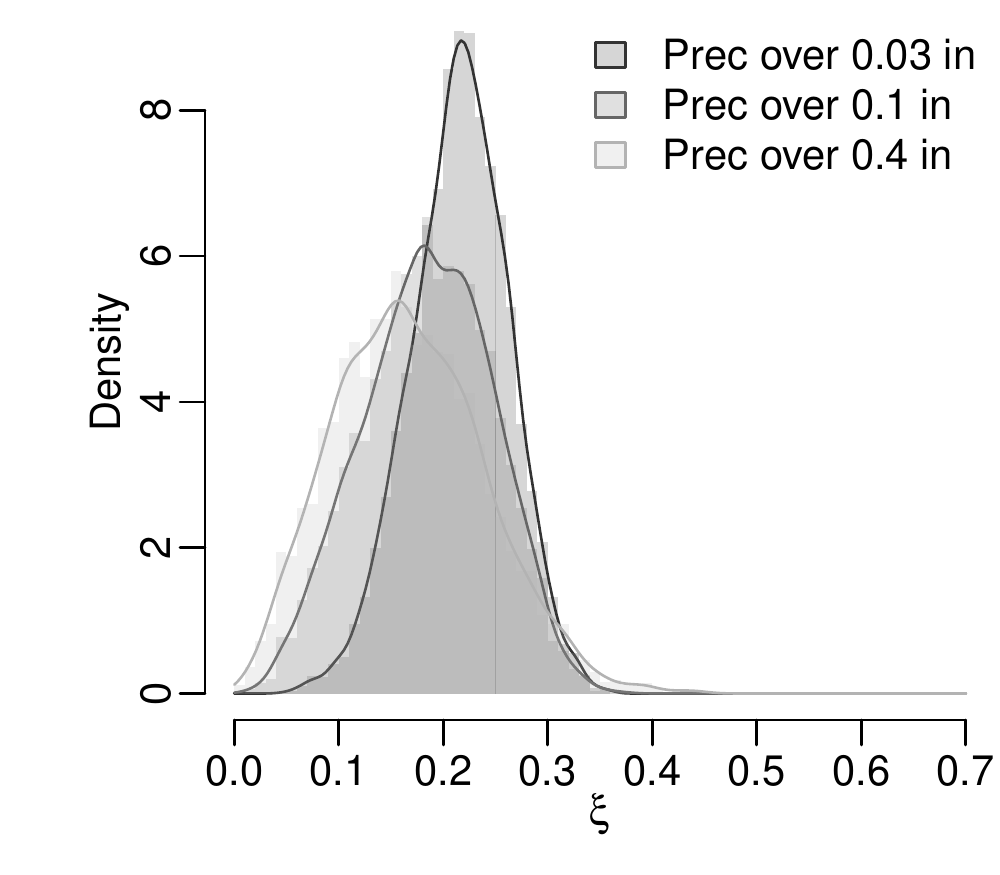}
\includegraphics[width=0.32\textwidth]{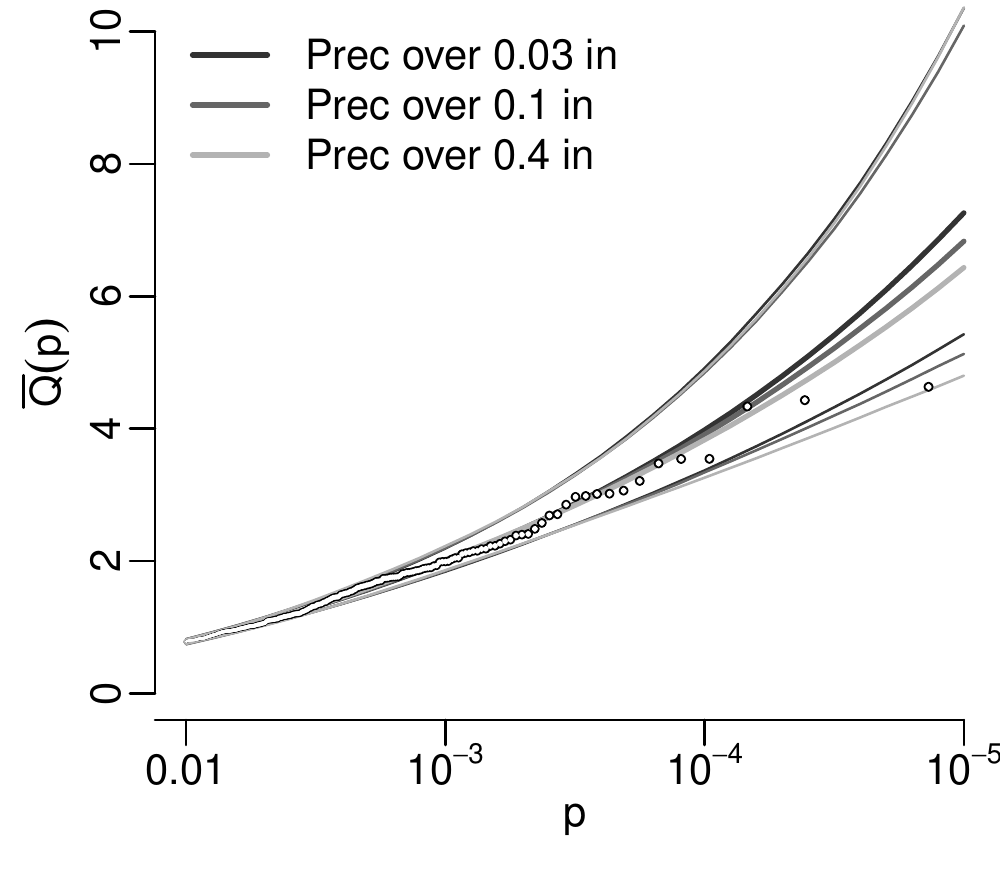}
\caption{Semiparametric estimation of tail heaviness ($\xi = \alpha^{-1}$) and high tail quantiles under further left truncation of the Fort Collins precipitation data. Posterior distribution of $\xi$ (left) widens and moves to the left slightly but maintaining overlap. High tail quantile estimates (right) remain robust. }
\label{fig:fc-sens}
\end{figure*}

\section{Concluding remarks}
\label{sec:discussion}
The semiparametric method analyzed here makes a case for likelihood based joint estimation of the bulk and the tail, with potential benefits that such joint estimation may improve estimation accuracy of high tail quantiles and provide a better uncertainty quantification of tail heaviness. Our asymptotic analysis reassures that sparse tail information does not get washed away by the bulk in such likelihood based estimation, however, suitable prior distributions are needed to strike a balance while also retaining full expressiveness of the bulk shape and tail decay rate. The transformation model \eqref{eq:model} appears to deliver the right theoretical platform especially when combined with the hierarchical LGP prior on the nonparametric density of the transformed data. A crucial element of the model is the choice of the Gaussian covariance kernel for the LGP prior. With a gamma prior on the inverse length-scale parameter of the kernel, the adaptive estimation accuracy of the LGP prior \citep{van2009adaptive} transfers seamlessly to our semiparametric setting.

{The semiparametric model \eqref{eq:model} adopts a GPD like tail and hence covers only the special Hall-Welsh class $\scD(\alpha, C_0, \xi, \epsilon, A)$ with $\xi = \min(1,\frac1{\alpha})$, albeit \cite{hall1984best, hall1985adaptive} make stringent assumptions on other quantities such as $\si(f)$ \citep{carpentier2015adaptive}. More ground might be recovered by using a more flexible parametric component, such as the three parameter extended-GPD family of \cite{beirlant2009second}. It could also be feasible to sharpen our posterior contraction rate in Theorem \ref{thm:tail-consis} by utilizing test functions that specifically exploit the idealized shape. {Theorem \ref{thm:tail-consis} intimately connects tail index estimation rate with the smoothness level of $f$. It will be interesting to examine whether this connection is intrinsic to the statistical task or simply an artifact of the proof technique adopted here.} We leave these extensions to a future study.}

{In applying the methodology developed here, an important consideration is whether one should fit the semiparametric model to the whole dataset, or only to data to the right of a low threshold. Our analysis of Fort Collins precipitation data indicates that while estimates are robust when data is truncated at or slightly over 0.02 inches, the estimates are sensitive to the presence of a massive number of excess zeros as well as a relative over-abundance of measurements at 0.01 and 0.02 which cannot be fully addressed by a simple jittering operation. Theorem \ref{thm:tail-consis} sheds light on this issue. Critical to the success of the joint estimation is the assumption of smoothness of the entire density function. In applications, it may be useful to threshold the data at a point above which the density function is believed to maintain a common level of smoothness. An alternative approach will be to apply a suitable smooth jitter to the data in the lower tail.}

A full estimation of the density function is also appealing with respect to model extension, e.g., in accounting for serial correlation or incorporating covariate information. For the latter task, we note that the transformation based density estimation model investigated here is closely related to the joint quantile regression model of \cite{yang2017joint}. Let $\zeta = \Psi^{-1}$ denote the quantile function of the transformed data $U = G_\theta(Y)$. Then the quantile function $Q(p)$ of the original data $Y$ could be expressed as $Q(p) = Q_\theta(\zeta(p)) = \int_0^{\zeta(p)} q_\theta(u) du$ where $Q_\theta = G_\theta^{-1}$ and $q_\theta = \dot Q_\theta$. To accommodate a predictor vector $x \in \bbR^d$, consider a quantile regression formulation
$Q(p|x) = \int_0^{\zeta(p)} q_\theta(u) \{1 + x^\top h(\omega(u))\}du$
where $\omega:u \mapsto (\omega_1(u), \ldots, \omega_d(u))^\top \in \bbR^d$ is unknown and $h(b)$ is a suitably chosen, fixed transformation that ensures $1 + x^\top h(b) \ge 0$ for all $b \in \bbR^d$ and all $x$ within a given bounded convex domain. This formulation is a special case of the joint linear quantile regression model proposed in \citet{yang2017joint} who jointly estimate $(\theta, \zeta, \omega)$ by adopting a hierarchical LGP prior on the quantile density $\dot\zeta$ and smooth Gaussian process priors on  $\omega_1, \ldots, \omega_d$. The theoretical analysis presented in the current paper is likely to yield a sharper understanding of asymptotic properties of the method by \cite{yang2017joint}, especially with respect to tail estimation. 

{
\section*{Appendix}
We present here proofs of the main results stated in Section \ref{sec:asymp-prop}. Several auxiliary technical results are stated whose proofs may be found in supplementary material. Several arguments build upon \cite{van2009adaptive} which we abbreviate below as VZ09.

\subsection*{A Auxiliary results for density estimation} Mixed partial derivatives of functions $\ell(\theta)$, $\theta = (\alpha, \sigma)$, are denoted by  $D^k\ell: = \frac{\partial^{|k|}\ell(\alpha,\sigma)}{\partial \alpha^{k_1} \partial \sigma^{k_2}}$ for any bi-index $k = (k_1, k_2) \in \{0,1,2, \ldots \}^2$ of order $|k| := k_1 + k_2$. Below $\Theta = [\alphalo,\alphahi]\times[\sigmalo,\sigmahi]$ with $0 < \alphalo < \alphahi < \infty$, $0 < \sigmalo < \sigmahi < \infty$ and any constants appearing in statements and proofs may implicitly depend on the boundary values of $\Theta$. Let $\Thetao$ 
denote the interior of $\Theta$. 

\begin{lem}
\label{lem:gpd}
Fix a density $\psi=\scL(\omega)$ with $\omega \in C^2[0,1]$. Let $q_\theta = p_{\theta, \psi}$, $\theta \in \Theta$. Then
$\max_{|k|=1} \sup_{\theta \in \Theta} |D^k \log q_\theta(y)| \le c_0\|\omega\|_{C^2}+ c_1\log(1+y)$, $\max_{|k| = 2} \sup_{\theta \in \Theta} |D^k \log q_\theta(y)| \le c_2\|\omega\|_{C^2}$,
for some constants $c_0,c_1,c_2$.
\end{lem}

\begin{lem}
\label{lem:taylor}
Fix a density $\psi = \scL(\omega)$ with $\omega \in C^2[0,1]$
.  If $\theta_1, \theta_2 \in \Thetao$ then
\begin{enumerate}
\item $d_{\kl}(p_{\theta_1,\psi}, p_{\theta_2,\psi}) \le c_2\|\omega_1\|_{C^2}\|\theta_1 - \theta_2\|^2$. 
\setcounter{lem5counter}{\value{enumi}}
\end{enumerate}
Moreover, there exist positive numbers $c_3, t_0$ such that if $\|\theta_1 - \theta_2\| < t_0$ then
\begin{enumerate}
\setcounter{enumi}{\value{lem5counter}}
\item $\int p_{\theta_2,\psi}(y) (\frac{p_{\theta_1,\psi}(y)}{p_{\theta_2,\psi}(y)} - 1)^2 dy \le c_3e^{3t_0c_0\|\omega\|_{C^2}} \|\theta_1 - \theta_2\|^2$, and
\item $V(p_{\theta_1,\psi}, p_{\theta_2,\psi}) \le c_3e^{3t_0c_0\|\omega\|_{C^2}}\|\theta_1 - \theta_2\|^2$.
\end{enumerate} 
\end{lem}


\begin{lem}
\label{lem:taylor2}
Fix $\theta_1\in \Thetao$, $\psi_1 = \scL(\omega_1)$ with $\omega_1 \in C^2[0,1]$ and $\epsilon \in (0,t_0)$. There exists a constant $K$ depending on $\|\omega_1\|_{C^2}$ such that 
$d_{\kl}(p_{\theta_1,\psi_1}, p_{\theta_2,\psi_2}) \le K\epsilon^2, 
V(p_{\theta_1,\psi_1},p_{\theta_2,\psi_2}) \le K\epsilon^2,$
for every $\theta_2 \in \Thetao$ with $\|\theta_1 - \theta_2\| \le \epsilon$ and every $\psi_2 = \scL(\omega_2)$ with $\|\omega_1 - \omega_2\|_\infty < \epsilon$. 
\end{lem}

\begin{lem}
\label{lem:h}
If $\theta_1, \theta_2 \in \Thetao$ and $\psi_1 = \scL(\omega_1)$, $\psi_2 = \scL(\omega_2)$ with $\omega_1, \omega_2 \in C^2[0,1]$ then 
$d_H(p_{\theta_1,\psi_1}, p_{\theta_2,\psi_2}) \le c_2\|\omega_1\|^{1/2}_{C^2} \|\theta_1 - \theta_2\| + \|\omega_1 - \omega_2\|_\infty e^{\|\omega_1 - \omega_2\|_\infty/2}.$

\end{lem}


\subsection*{B Proof of Condition C1*}
For a mean-zero Gaussian process on $[0,1]$ with covariance $c_\lambda$, let $\nu^\lambda$ denote the Gaussian measure with respect to the Borel $\sigma$-algebra on $(C[0,1], \|\cdot\|_\infty)$; see Section 2 of \cite{van2008rates} for necessary technical details. Define $\bar\nu(\cdot) = \int \nu^\lambda(\cdot) \pi_\lambda(\lambda)d\lambda$ as the probability measure on $C[0,1]$ under the hierarchical Gaussian process prior specification with a prior density $\pi_\lambda$ on the inverse length-scale parameter. In light of Lemma \ref{lem:taylor2}, to prove Condition C1* it is enough to to show that for all large $n$, both $t_n=\pi_\theta(\{\theta:\|\theta - \theta^*\| \le \bar\epsilon_n\})$ and $w_n=\bar\nu(\{\omega : \|\omega - \phi^*\|_\infty \le \bar\epsilon_n\})$ are larger than $e^{-CKn\bar\epsilon_n^2}$ where $C$ is a constant that may depend on $\|\phi^*\|_{C^2}$. The bound on $t_n$ follows trivially from the assumption on $\pi_\theta$ and that on $w_n$ follows directly from Theorem 3.1 of VZ09.

\subsection*{C Proof of Condition C2*}
Suppose there exist sets $\scB_1,\scB_2,\ldots \subset C^2[0,1]$ such that for all large $n$,
\begin{align}
\bar\nu(\scB_n^c) & \le e^{-(C+4)n\bar\epsilon_n^2}\label{eq:sieve1}\\
\log N(\bar\epsilon_n, \scB_n, \|\cdot\|_{C^2}) & \le n \epsilon_n^2/2\label{eq:sieve2}\\
\sup\{\|\omega\|_{C^2}: \omega \in \scB_n\} & \le (n\bar\epsilon_n)^b,\label{eq:sieve3}
\end{align}
for some $b \ge 1$. 
Then, $\scF_1, \scF_2, \ldots$ could be simply constructed as $\scF_n = \{p_{\theta,\psi}: \theta \in \Theta, \psi =\scL(\omega), \omega \in \scB_n\}$. To see that these sets satisfy the requirements of C2*, notice that $\Pi(\scF_n^c) = \bar\nu(\scB_n^c) \le e^{-(C+4)n\bar\epsilon_n^2}$, and,  by Lemma \ref{lem:h}, $\log N(\bar\epsilon_n, \scF_n, d_H) \le \log N(\bar\epsilon_n/(c_2(n\bar\epsilon_n)^{b/2}), \Theta, \|\cdot\|) + \log N(\bar\epsilon_n, \scB_n, \|\cdot\|_{C^2}) \le \log(\diam(\Theta)\cdot n^b) + n\epsilon_n^2/2 \le n\epsilon_n^2$ for all large $n$. 

Conditions \eqref{eq:sieve1}-\eqref{eq:sieve2} mirror conditions (3.6)-(3.7) of VZ09, but with the crucial technical difference that we need entropy calculation in $\|\cdot\|_{C^2}$ as opposed to $\|\cdot\|_\infty$. Accordingly, we adapt the construction of $\scB_n$ for $(C[0,1],\|\cdot\|_\infty)$ by VZ09 to $(C^2[0,1], \|\cdot\|_{C^2})$. Our adaptation also produces a smaller exponent $b$ in \eqref{eq:sieve3} than what is possible with the original construction of VZ09. Although a smaller exponent is not critical to the current proof, it proves useful for tail index estimation. Our adaptation builds on the well known fact that a centered Gaussian process with covariance $c_\lambda$ has infinitely differentiable sample paths with probability one. Therefore the Gaussian measures $\nu^\lambda$ introduced in the preceding section could also be viewed as probability measures with respect to the refined Borel $\sigma$-algebra of $(C^2[0,1], \|\cdot\|_{C^2})$. A more formal treatment is outlined below. Hereafter, $\Re(z)$ and $\Im(z)$ denote the real and imaginary parts of a complex number $z$.

VZ09 show that the reproducing kernel Hilbert space $\bbH^\lambda$ associated with $c_\lambda$ consists of functions $h(u)=\Re(\int e^{ut\sqrt{-1}} \eta(t)\mu_\lambda(t)dt)$ with $\|h\|_{\bbH^\lambda} = \|\eta\|_{L_2(\mu_\lambda)}$, where $\mu_\lambda(t) = e^{-t^2/{4\lambda^2}}/(2\lambda\sqrt{\pi})$ is the spectral density associated with $c_\lambda$. By applying Cauchy-Schwarz inequality, with differentiations under integration as needed, it follows that 
\begin{equation}
\|h\|_\infty \le \|h\|_{\bbH^\lambda}, \|\dot h\|_\infty \le \sqrt{2}\lambda\|h\|_{\bbH^\lambda},~\mbox{and}~\|\ddot h\|_\infty \le \sqrt{12}\lambda^2\|h\|_{\bbH^\lambda}.
\label{eq:bounds}
\end{equation} 
Clearly, $\|h\|_{C^2} \le (1+\sqrt{2}\lambda+\sqrt{12}\lambda^2)\|h\|_{\bbH}$ and $\bbH^\lambda$ can be continuously and densely embedded within the Banach space $(C^2[0,1], \|\cdot\|_{C^2})$, guaranteeing a Borel measure $\nu^\lambda$ on the embedding Banach space matching the law of a centered Gaussian process with covariance $c_\lambda$. As before, define $\bar\nu(\cdot) = \int \nu^\lambda(\cdot) \pi_\lambda(\lambda)d\lambda$.

Let $\bbH^\lambda_1$ and $\scB_1$ denote the unit balls of $\bbH^\lambda$ and $C^2[0,1]$. Recall, $\bar\epsilon_n = B n^{-t}(\log n)^s$, $\epsilon_n = \bar\epsilon_n \log n$ where we are free to choose $B>0$. To start off, take $B$ large enough such that $r_n = n\bar\epsilon_n^2 > 1$ for all $n$. Let $m_n$ be the smallest integer larger than $\log_2(r_n)$. Define
\begin{equation}
\label{eq:Bn}
\scB_n = [\{\cup_{j = 1}^{m_n} (\sqrt{2}M_n\bbH^{2^j}_1)\}\cup \{M_n\delta_n^{-1/2}\bbH^1_1\} \cup \{\cup_{\lambda < \delta_n} (M_n \bbH_1^{\lambda})\} ]+ \bar\epsilon_n \bbB_1
\end{equation}
where $M_n = 16Cr^{1/2}_n\log(r_n/\bar\epsilon_n)$ with $C$ taken from Lemma \ref{lem:rkhs-entropy} below, and $\delta_n = \bar\epsilon_n/(4M_n)$. Because of \eqref{eq:bounds}, for all large $n$, $\scB_n \subset 5r^2_nM_n\bbB_1$ and hence $\scB_n$ satisfies \eqref{eq:sieve3} with $b = 2.5$. By Lemma 4.7 of VZ09, $\scB_n \supset M_n\bbH^\lambda + \bar\epsilon_n \bbB_1$ for every $0 < \lambda \le r_n$. Borell's inequality implies that $\nu^\lambda(\scB_n^c) \le 1 - \Phi(\Phi^{-1}(\nu^\lambda(\bar\epsilon_n \bbB_n)) + M_n) \le 1 - \Phi(\Phi^{-1}(\nu^{r_n}(\bar\epsilon_n \bbB_1)) + M_n)$ where the second inequality follows since $\nu^\lambda(\epsilon \bbB_1)$ is decreasing in $\lambda$ for every $\epsilon > 0$ (Lemma \ref{lem:small-ball-ordering} below). As $\nu^{r_n}(\bar\epsilon_n \bbB_1)\le \nu^1(\bar\epsilon_n \bbB_n) < 1/4$ and $M_n \ge 4 \sqrt{\log(1/\nu^{r_n}(\bar\epsilon_n \bbB_1))}$ for all large $n$ (Lemma \ref{lem:rkhs-entropy} below), it must be that $\nu^\lambda(\scB_n^c) \le 1 - \Phi(M_n/2) \le e^{-M_n^2/8} \le e^{-r_n}$ for every $\lambda \in (0, r_n)$, for all large $n$. This establishes \eqref{eq:sieve1}, with $B$ chosen suitably large, since $\pi_\lambda((r_n, \infty)) \le e^{-C_3r_n}$ for all large $n$ for some constant $C_3$.

To establish \eqref{eq:sieve2}, first note that every $h \in \cup_{\lambda < \delta_n}(M_n \bbH^\lambda_1)$ satisfies $\|h - h(0)\|_{C^2} \le \bar\epsilon_n$ by \eqref{eq:bounds}, i.e., as an element of $C^2[0,1]$, the function $h(u)$ is within $\bar\epsilon_n$ distance of a constant function whose constant value ranges within $[-M_n, M_n]$. Clearly, $\log N(2\bar\epsilon_n, \cup_{\lambda < \delta_n} (M_n \bbH^\lambda_1) + \bar\epsilon_n \bbB_1, \|\cdot\|_{C^2}) \le \log\frac{2M_n}{\bar\epsilon_n}$. Next, by Lemma \ref{lem:rkhs-entropy} below, $\log N(2\bar\epsilon_n, M_n\delta_n^{-1/2}\bbH^1_1 + \bar\epsilon_n \bbB_1, \|\cdot\|_{C^2})\le C\log^2(\frac{M_n}{\bar\epsilon_n}\delta_n^{-1/2})$ and $\log N(2\bar\epsilon_n, \sqrt{2}M_n \bbH^{2^j}_1 + \bar\epsilon_n \bbB_1, \|\cdot\|_{C^2}) \le C2^j \log^2(\frac{2^{j+1/2}M_n}{\bar\epsilon_n}) \le 2Cr_n \log^2(\frac{r_nM_n}{\bar \epsilon_n})$ for each $1 \le j \le m_n$ by the monotonicity of $\log x$. Consequently,  
\[
\log N(2\bar\epsilon_n, \cup_{j = 1}^{m_n} (\sqrt{2}M_n\bbH^{2^j}_1) + \bar\epsilon_n \bbB_1, \|\cdot\|_{C^2}) \le \log(m_n) + 2Cr_n(\log \tfrac{r_nM_n}{\bar \epsilon_n})^2,
\]
concluding the proof of Condition C2*. Two auxiliary results used in the above prove are:


\begin{lem}
\label{lem:small-ball-ordering}
For any fixed $\epsilon > 0$, the small ball probability $\nu^\lambda(\epsilon \bbB_1)$ is decreasing in $\lambda > 0$. 
\end{lem}

\begin{lem}
\label{lem:rkhs-entropy}
There exist $C, \epsilon_0$ such that for all $\lambda \ge 1$ and all $\epsilon<\epsilon_0$,
(a) $\log N(\epsilon, \bbH^\lambda_1, \|\cdot\|_{C^2}) \le C \lambda \log^2 (\lambda/\epsilon)$, and
(b) $-\log\nu^\lambda(\epsilon \bbB_1) \le C \lambda \log^2(\lambda/\epsilon)$.
\end{lem}

\subsection*{D Auxiliary results for tail estimation}
If $f$ is heavy tailed then $\lim_{y \to \infty} |{y^{\alpha_+(f)}} \bar F(y)/{\si(f)} - 1| = 0$. For our semiparametric analysis it is useful to consider classes of heavy tailed densities for which this convergence holds uniformly. Define $\textstyle \scT(t, \delta) = \{f: \sup_{y \ge t} |{y^{\alpha_+(f)}} \bar F(y)/{\si(f)} - 1| \le \delta\}$, for any arbitrary $t > 0, \delta > 0$. For any $f \in \scF$, let $\bbP^n_f$ denote the joint probability law of $(Y_1, \ldots, Y_n)$ with $Y_i \sim f$ independently of one another and $\bbP^n_f h$ denote expectation of $h(Y_1,\ldots,Y_n)$ under $\bbP^n_f$. For the following lemma, let $f^*$ denote an arbitrary heavy tailed density with $\alpha^* = \ti(f^*) \in (\alphalo, \alphahi)$ and let $\epsilon_n \to 0$ be an arbitrary positive sequence satisfying $n\epsilon_n^2 \to \infty$. By a test function we mean any statistic that takes values in $[0,1]$. 
\begin{lem}
\label{lem:tests}
Suppose there exist positive sequences $t_n \to \infty$, $\delta_n \to 0$ such that $f^* \in \scT(t_n, \delta_n)$ and $\min\{\bar F^*(t_n), \bar F^*(t_n)^{1/2}\delta_n\} \ge 3 \epsilon_n$  for all large $n$. Then there exist test functions $T_n = T_n(Y_1,\ldots, Y_n)$ satisfying $\bbP^n_{f^*}T_n \le 4e^{-n\epsilon_n^2}$ and $\sup\{\bbP^n_f(1 - T_n): f \in \scT(t_n, \delta_n), \ti(f) < \alphahi, |\ti(f) - \alpha^*| > 2^{4+\alphahi} \delta_n\} \le 4e^{-n\epsilon_n^2}$ for all large $n$. 
\end{lem}

\begin{lem}
\label{lem:pareto}
Suppose $\tau_n \to 0, D_n \to \infty$ are positive sequences and $t_n = (D_n/\tau_n)^{1/\min(1, A)}$ for some $A > 0$. Then, with $B_1 >0$ chosen sufficiently large, $\{f = p_{\theta,\psi}: \theta = (\alpha,\sigma) \in \Theta, \alpha \ge A, \psi = \scL(\omega), \|\dot\omega\|_\infty \le D_n\} \subset  \scT(t_n, B_1\tau_n)$ for all large $n$.
\end{lem}


\subsection*{E Proof of Theorem \ref{thm:tail-consis}}

Our argument is based on the proof of Theorem 8.9 in \cite{ghosal2017fundamentals}. Consider again the sets $\scF_n = \{p_{\theta,\psi}: \theta \in \Theta, \psi =\scL(\omega), \omega \in \scB_n\}$  from the proof of Condition C2* where $\scB_n$ is as in \eqref{eq:Bn} with $\bar \epsilon_n = B\{(\log n)^2 / n\}^{\frac{\beta}{2\beta + 1}}$ for some large $B$. Recall that $\Pi(\scF_n^c) \le e^{-(C+4)n\bar\epsilon^2}$ for some constant $C$. Define 
\[
\scU_n = \{p_{(\theta,\sigma),\psi}:\theta = (\alpha,\sigma) \in \Theta, |\alpha - \alpha^*| > B_1n^{-\rho}(\log n)^s, \psi=\scL(\omega), \omega \in C^2[0,1]\}.
\]
It follows from Bayes' formula for $\Pi(\scU_n \mid Y_1, \ldots, Y_n)$ that with $A_n:=\{(y_1,\ldots,y_n): \int_\scF \prod_{i=1}^n \frac{f(y_i)}{f^*(y_i)}\Pi(df) \ge e^{-(2 + C)n\bar\epsilon_n^2}\}$ and for any test function $T_n:\bbR^n \to [0,1]$,
\[
\bbP^n_{f^*}\Pi(\scU_n \mid Y_1,\ldots, Y_n) \le \bbP^n_{f^*}T_n + \bbP^n_{f^*}(A_n^c) + e^{(2 + C)n\bar\epsilon_n^2}\left[\sup_{f \in \scF_n \cap \,\scU_n} \bbP^n_f(1 - T_n) + \Pi(\scF_n^c)\right]
\]
Now, $\lim_{n\to \infty} e^{(2 + C)n \bar \epsilon_n^2} \Pi(\scF_n^c) = 0$ by construction and $\lim_{n \to \infty}\bbP^n_{f^*}(A_n^c) = 0$ by Lemma 8.10 of \cite{ghosal2017fundamentals}. Therefore the proof of the theorem is complete once we have shown the existence of test functions $(T_n: n \ge 1)$ satisfying
\begin{equation}
\label{eq:tests}
\lim_{n \to \infty} \bbP^n_{f^*} T_n = 0, \quad \sup_{f \in \scF_n \cap \,\scU_n} \bbP^n_f (1 - T_n) \le e^{-(4 + C)n\bar\epsilon_n^2}~\mbox{for all large}~n.
\end{equation}
We shall construct such a test function based on Lemmas \ref{lem:tests} and \ref{lem:pareto}.

Take $\epsilon_n = (4 + C)^{1/2}\bar\epsilon_n$.
For any $f = p_{\theta, \psi} \in \scF_n$ it follows from \eqref{eq:bounds} that if $\phi = \log \psi$ then $\|\dot \phi\|_\infty \le D_n := C_1 r_n^{3/2}\log n= C_1 n^{\frac32(1 - 2\gamma)}(\log n)^{6\gamma + 1}$ for some constant $C_1$. Set $\alpha_1 = \xi \alpha^*$ and note that  $\alphalo < \alpha_1 < \min(1,\alpha^*)$ and partition $\scU_n = \scU_{1n} \cup \scU_{2n}$ where $\scU_{1n} = \scU_n \cap \{f: \alphalo \le \ti(f) < \alpha_1\}$ and $\scU_{2n} = \scU_n \cap \{f: \alpha_1 \le \ti(f) \le \alphahi\}$. By Lemma \ref{lem:pareto} (with $A = \alphalo$), for any $\rho_1, s_1 > 0$, 
\begin{equation}
\label{eq:U1}
\scF_n \cap \scU_{1n} \subset \scT(t_{1n}, \delta_{1n}) \cap \{f:\ti(f) \le \alphahi, |\ti(f) - \alpha^*| > 2^{4 +  \alphahi}\delta_{1n}\}~\mbox{for all large}~n, 
\end{equation} 
where $\delta_{1n} = B_{12}\taulo_{n}$, $\taulo_{n} = C_{12}n^{-\rho_1}(\log n)^{s_1}$, $t_{1n} = (D_n/\taulo_{n})^{1/\alphalo}$ and $B_{12}, C_{12}$ are large constants to be adjusted. We next show that $\rho_1, s_1 > 0$ could be chosen so that 
\begin{equation}
\label{eq:Fbar1}
\min\{\bar F^*(t_{1n}), \delta_{1n} \bar F^*(t_{1n})^{1/2}\} \ge 3\epsilon_n~\mbox{for all large}~n.
\end{equation} 
Indeed, $\bar F^*(t_{1n}) \ge \tfrac12\zeta(f^*)t_{1n}^{-\alpha^*} = \tfrac12\zeta(f^*)(\tfrac{C_{12}}{C_1})^{1/\xilo}\times n^{-\{\rho_1 + \frac32(1 - 2\gamma)\}/\xilo}(\log n)^{(s_1 - 6\gamma - 1)/\xilo}$ for all large $n$, 
where $\xilo = \alphalo / \alpha^* \in (0,1)$. Therefore, with a suitably large choice of ${C}_{12}$ we can make $\bar F^*(t_{1n}) \ge 3\epsilon_n$ for all large $n$ provided $\rho_1 \le \hat\rho(\xilo)$, and in case of an equality, $s_1 = 
2\rho_1 + 4$. 
On the other hand, in order to have $\delta_{1n} \bar F^*(t_n)^{1/2} \ge 3 \epsilon_n$, we need to choose ${B}_{12}$ suitably large and $\rho_1 \le \bar\rho(\xilo)$, and in case of an equality, $s_1 = 
2\rho_1 + \frac{4}{\alpha^*(2\xilo+1)}$. 
With \eqref{eq:U1}-\eqref{eq:Fbar1} established with $\rho_1 > 0$ chosen as the minimum of the above two bounds and $s_1 > 0$ set accordingly, apply Lemma \ref{lem:tests} to conclude that there exist test functions $T_{1n} = T_{1n}(Y_1,\ldots,Y_n)$ such that $\bbP^n_{f^*}T_{1n} \le e^{-n\epsilon_n^2}$ and $\sup\{\bbP^n_f(1 - T_{1n}): f \in \scF_n \cap \scU_{1n}\} \le e^{-n\epsilon_n^2}$ for all large $n$. 

Next we repeat the same arguments for testing $f = f^*$ versus $f \in \scF_n \cap \scU_{2n}$. Rewrite the target rate as $B_1 n^{-\rho} (\log n)^s=2^{4 + \alphahi}\delta_n$ where $\delta_n = B_{22} \tau_n$, $\tau_n = C_{22}n^{-\rho}(\log n)^s$, and $t_n = (D_n/\tau_n)^{1/\alpha_1}$, with $B_{22}, C_{22}$ to be adjusted as needed. As argued in the preceding paragraph, the choices of $\rho, s$ imply that $\min(\bar F^*(t_n), \delta_n \bar F^*(t_n)^{1/2}) \ge 3\epsilon_n$ and Lemma \ref{lem:pareto} (with $A = \alpha_1$) implies that $\scF_n \cap \scU_{2n} \subset \scT(t_{n}, \delta_{n}) \cap \{f : \ti(f) \le \alphahi, |\ti(f) - \alpha^*| > \delta_n\}$. Therefore, by Lemma \ref{lem:tests}, there are test functions $T_{2n} = T_{2n}(Y_1,\ldots,Y_n)$ such that $\bbP^n_{f^*}T_{2n} \le e^{-n\epsilon_n^2}$ and $\sup\{\bbP^n_f(1 - T_{2n}): f \in \scF_n \cap \scU_{2n}\} \le e^{-n\epsilon_n^2}$ for all large $n$. The proof is now complete by taking $T_n = \max(T_{1n}, T_{2n})$. }

\newpage
\begin{center}
{\large\bf SUPPLEMENTARY MATERIAL}
\end{center}

{
\section*{Proofs of auxiliary results}

\begin{proof}[Proof of Lemma 5]
Clearly $\phi = \log \psi \in C^2[0,1]$ with $\dot\phi = \dot\omega$, $\ddot\phi = \ddot\omega$.
Let $\nab\theta$ and $\nab\theta^2$ denote the first and second order vector differential operators with respect to $\theta = (\alpha, \sigma)$. Then, 
\begin{align*}
\nab\theta \log q_\theta(y) &= \nab\theta\log g_\theta(y) + \dot\phi(G_\theta(y)) \nab\theta G_\theta(y)\\
\nab\theta^2 \log q_\theta(y) & = \nab\theta^2 \log g_\theta(y)  + \dot \phi(G_\theta(y)) \nab\theta^2 G_\theta(y) + \ddot \phi(G_\theta(y)) \nab\theta G_\theta(y) \nab\theta G_\theta(y)^\top
\end{align*}
which immediately proves the result because $\del\alpha \log g_\theta(y)$ is bounded by a shifted and scaled version of $\log(1 + y)$, and $\del\sigma \log g_\theta(y)$ as well as every term in $\nab\theta^2 \log g_\theta(y)$, $\nab\theta G_\theta(y)$ and $\nab\theta^2 G_\theta(y)$ is uniformly bounded over $y \ge 0$ and $\theta \in \Theta$. For completeness we list below the first and second order partial derivatives of $\log g_\theta(y)$ and $G_\theta(y)$; expressed in terms of $z = (1 + \frac{y}{\alpha\sigma})^{-1} \in (0,1]$,
\begin{align*}
&\textstyle \del{\alpha}\log g_\theta(y) = \log z + \frac{1-z}{\alpha}, \del\sigma \log g_\theta(y) = \frac{\alpha-(\alpha+1)z}{\sigma},\\
&\textstyle \deltwo{\alpha} \log g_\theta(y) = \textstyle\frac{(1-z)\{\alpha - 1 - (\alpha + 1)z\}}{\alpha^2},
\textstyle \deltwo{\sigma} \log g_\theta(y) = \textstyle \frac{(\alpha+1)z^2 - \alpha}{\sigma^2}, 
\textstyle \deltwom{\alpha}{\sigma} \log g_\theta(y) = \textstyle \frac{\{\alpha - (\alpha+1)z\}(1-z)}{\alpha\sigma}\\
&\textstyle \frac{\partial G_\theta(y)}{\partial\alpha} = \textstyle  (\log z + 1 - z)z^\alpha,
\textstyle \frac{\partial^2 G_\theta(y)}{\partial\alpha^2} = \textstyle \{(\log z + 1 - z)^2 + \frac{(1-z)^2}{\alpha}\}z^\alpha\\
&\textstyle\frac{\partial G_\theta(y)}{\partial\sigma} = \frac{\alpha(1-z)}{\sigma}z^\alpha,
\frac{\partial^2G_\theta(y)}{\partial\sigma^2} = \frac{\alpha(1-z)\{\alpha -1-z(\alpha+1)\}}{\sigma^2}z^\alpha, 
\frac{\partial^2G_\theta(y)}{\partial\alpha\partial\sigma} = \frac{(1-z)\{\alpha(\log z + 1 - z) + 1 - z\}}{\sigma} z^\alpha.
\end{align*}

\end{proof}

\begin{proof}[Proof of Lemma 6]
Denote $q_\theta = p_{\theta,\psi}$, $\theta \in \Theta$. By Taylor's theorem, for $\theta, \theta + u$ in the interior of $\Theta$,
\[
\log \frac{q_{\theta + u}(y)}{q_{\theta}(y)} = R_1(\theta,u,y) = u^\top \nabla_\theta \log q_\theta(y) + R_2(\theta,u,y)
\]
with $|R_j(\theta,u,y)| \le \|u\|^j\max_{|k| = j} \sup_{\theta \in \Theta} |D^k \log q_\theta(y)|$, $j = 1,2$. The first claim now follows because $d_\kl(q_\theta, q_{\theta+u})=\int q_\theta(y) \log \frac{q_\theta(y)}{q_{\theta + u}(y)} dy = 0 + \int R_2(\theta,u,y) q_\theta(y) dy \le c_2\|\omega\|_{C^2} \|u\|^2$ by Lemma 5
. Next, use the inequality $|e^x - 1|\le |x|e^{|x|}$ to conclude 
\[
\textstyle \left|\frac{q_{\theta+u}(y)}{q_{\theta}(y)} - 1\right| 
\le |R_1(\theta,u,y)|e^{|R_1(\theta,u,y)|} \le \|u\|\{c_0\|\omega\|_{C^2} + c_1\log(1+y)\}e^{ \|u\|\{c_0\|\omega\|_{C^2} + c_1\log(1+y)\}}
\]
by Lemma 5
. Therefore, $\int q_\theta(y) (\frac{q_{\theta + u}(y)}{q_{\theta}(y)} - 1)^2 dy \le c_4 \|u\|^2$ where
\[
c_4 =\sup_{\theta \in \Theta}\int \{c_0\|\omega\|_{C^2} + c_1\log(1+y)\}^2e^{2t_0\{c_0\|\omega\|_{C^2} + c_1\log(1+y)\}} q_\theta(y)dy \le c_3e^{3t_0c_0\|\omega\|_{C^2}}
\]
with $c_3 := t_0^{-2}\sup_{\theta \in \Theta}\int(1 + y)^{3t_0c_1} q_\theta(y) dy$ a finite number if $t_0 < \alphalo/(3c_0)$. This proves the second claim as well as the third claim since $V(q_\theta, q_{\theta+u}) = \int R_1(\theta,u,y)^2 q_\theta(y)dy \le c_4\|u\|^2$.
\end{proof}


\begin{proof}[Proof of Lemma 7]
Denote $p_{ij} = p_{\theta_i, \psi_j}$, $P_{ij}[g] := \int g(y) p_{ij}(y) dy$, for $i,j \in \{1,2\}$. Note that $d_\kl(p_{11},p_{22}) = d_{\kl}(p_{11},p_{21}) + P_{11} [\log \frac{p_{21}}{p_{22}}] \le c_2\|\omega_1\|_{C^2}\|\theta_1 - \theta_2\|^2 + P_{11} [\log \frac{p_{21}}{p_{22}}] $ by Lemma 6
. Use the fact that every $p_{ij}$ has full support on $[0,\infty)$ to write 
\[
\textstyle P_{11}[\log \frac{p_{21}}{p_{22}}] = P_{21}[\frac{p_{11}}{p_{21}} \log \frac{p_{21}}{p_{22}}] = P_{21}[(\frac{p_{11}}{p_{21}} -1)\log \frac{p_{21}}{p_{22}}] + d_\kl(p_{21},p_{22}).
\]
Notice, $d_\kl(p_{21},p_{22}) = d_{\kl}(\psi_1,\psi_2) \le K_0\epsilon^2$ for some constant $K_0$ that depends only on $t_0$; see Lemma 3.1 of \cite{van2008rates}. An application of Cauchy-Schwarz inequality gives
\[
\textstyle P_{21}[(\frac{p_{11}}{p_{21}} -1)\log \frac{p_{21}}{p_{22}}] \le \{P_{21}[(\frac{p_{11}}{p_{21}} -1)^2]\}^{1/2}\{P_{21}[(\log \frac{p_{21}}{p_{22}})^2]\}^{1/2}.
\]
Clearly $P_{21}[(\log \frac{p_{21}}{p_{22}})^2] = V(\psi_1,\psi_2) \le \|\log\frac{\psi_1}{\psi_2}\|^2_\infty \le 4\|\omega_1 - \omega_2\|^2$, and, by Lemma 6
, $P_{21}[(\frac{p_{11}}{p_{21}}-1)^2] \le c_3e^{3t_0c_0\|\omega_1\|_{C^2}} \|\theta_1-\theta_2\|^2$. Additionally, $V(p_{11},p_{22}) \le 2V(p_{11},p_{21}) + 2\|\log \frac{\psi_1}{\psi_2}\|_\infty^2 \le c_3e^{3t_0c_0\|\omega_1\|_{C^2}} \|\theta_1-\theta_2\|^2 + 4 \|\omega_1 - \omega_2\|^2_\infty$ by Lemma 6
. This concludes the proof of the lemma with $K =\max(4,K_0, c_2\|\omega_1\|_{C^2},c_3e^{3t_0c_0\|\omega_1\|_{C^2}})$. 
\end{proof}

%

\begin{proof}[Proof of Lemma 8]
Denote $p_{ij} = p_{\theta_i, \psi_j}$, $i,j \in \{1,2\}$. By triangle inequality, $d_H(p_{11}, p_{22})  \le d_H(p_{11}, p_{21}) + d_H(p_{21}, p_{22})$. The second term on the right equals $d_H(\psi_1,\psi_2)$ which is bounded by $\|\omega_1 - \omega_2\|_\infty \exp\{\|\omega_1 - \omega_2\|_\infty/2\}$ by Lemma 3.1 of \citet{van2008rates}. The desired bound on the first term follows by the inequality $d_H(p_{11}, p_{21}) \le d_{\kl}(p_{11},p_{21})^{1/2}$ and Lemma 6
.
\end{proof}

\begin{proof}[Proof of Lemma 9]
Let $W(t)$ be a centered Gaussian process on $\bbR$ with $\cov(W(s),W(t)) = e^{-(t - s)^2}$, $t,s\in\bbR$. Then $\nu^\lambda$ is the probability law of the rescaled process $W^\lambda=(W^\lambda(t):=W(\lambda t): 0 \le t \le 1)$. The proof is complete by noting that 
\[
\|W^\lambda\|_{C^2}=\sup_{0\le t\le\lambda}|W(t)| +\lambda \sup_{0\le t\le\lambda}|\dot W(t)| + \lambda^2\sup_{0\le t\le\lambda}|\ddot W(t)|
\]
where, with probability one, the right hand side is non-decreasing in $\lambda$. 
\end{proof}


\begin{proof}[Proof of Lemma 10]
Fix $\lambda \ge 1$ and $\delta < 1/12$. Recall that $\bbH^\lambda_1$ consists of functions $\Re(h_\eta)$ where $h_\eta(u) = \int e^{ut\sqrt{-1}}\eta(t)\mu_\lambda(t)$ with $\|\eta\|_{L_2(\mu_\lambda)} \le 1$. By applying Cauchy-Schwarz inequality, with differentiations under integration as needed, it follows that 
\begin{equation}
\|h\|_\infty \le 1, \|\dot h\|_\infty \le \sqrt{2}\lambda,~\mbox{and}~\|\ddot h\|_\infty \le \sqrt{12}\lambda^2.
\label{eq:bounds}
\end{equation}
Any such $h_\eta$ could be extended to an analytic function $h_\eta$ on the complex plane $\bbC$ such that $|\frac{d^j}{dz^j}h_\eta(z)| \le 8 \lambda^j e^{2|\Im(z)|^2\lambda^2}$, $z \in \bbC$ and $j \in \{0,1,2\}$. By Proposition C.9 of \cite{ghosal2017fundamentals}, there is a collection $\scP = \{P_1, \ldots, P_N\}$ of piecewise polynomials on $[0,1]$ with $\log N \le C_0\lambda (\log \frac\lambda\delta)^2$ such that every $h \in \bbH^\lambda_1$ satisfies $\|\ddot h - P_n\|_\infty < \delta$ for some $1 \le n \le N$; here $C_0$ is a universal constant. Consider an expanded collection $\tilde \scP$ of functions $\tilde P(u) = a + b u + \int_0^1 (u - t)_+ P(t)dt$ where $a$ belongs to a $\delta$-net of $[-1,1]$, $b$ belongs to a $\delta$-net of $[-\sqrt{2}\lambda,\sqrt{2}\lambda]$ and $P\in \scP$. Use \eqref{eq:bounds} 
and Taylor's Theorem (second order, with residual in the integral form) to conclude every $h \in \bbH^\lambda_1$ satisfies $\|h - \tilde P\|_{C^2} < 6\delta$ for some $\tilde P \in \tilde \scP$. This establishes the first claim because the cardinality $\tilde N$ of $\tilde \scP$ satisfies $\log \tilde N \le \log N + \log(2/\delta) + \log(2\sqrt{2}\lambda/\delta) \le C \lambda (\log \frac{\lambda}{6\delta})^2$ for all $\epsilon < 1/2$ and a new universal constant $C$. As shown in the proof of Lemma 4.7 of \cite{van2009adaptive}, the second claim follows as a corollary to the first claim and Theorem 2 of \cite{li1999approximation}.
\end{proof}

\begin{proof}[Proof of Lemma 11]
Let $S_n(t) = \sum_{i = 1}^n I(Y_i > t)$ denote the sample exceedance count over a threshold $t$. Define the test functions 
\[
T_{1n} = I(|\tfrac{S_n(t_n)}{n} - \bar F^*(t_n)| > \epsilon_n),\quad T_{2n} = I(|\tfrac{S_n(2t_n)}{\max\{S_n(t_n), 1\}} - \tfrac{\bar F^*(2t_n)}{\bar F^*(t_n)}| > \delta_n),
\]
and take $T_n = \max(T_{1n}, T_{2n})$. 
Since $T_n \le T_{1n} + T_{2n}$, we have $\bbP^n_{f^*}T_n \le \bbP^n_{f^*}T_{1n} + \bbP^n_{f^*}T_{2n} \le 2e^{-2n\epsilon_n^2} + \bbP^n_{f^*}[2e^{-2S_n(t_n)\delta_n^2}]$ by applications of Hoeffding's inequality where the second term is handled by the law of iterated expectation with an intermediate conditioning on $S_n(t_n)$. Now, for all large $n$, $\bbP^n_{f^*}[e^{-2S_n(t_n)\delta_n^2}] = [1 - \bar F^*(t_n)(1 - e^{-2\delta_n^2})]^n \le [1 - \bar F^*(t_n)\delta_n^2]^n \le e^{-n\bar F^*(t_n)\delta_n^2} \le e^{-9n\epsilon_n^2}$; the last two inequalities hold because $1 - e^{-2x} \ge x$ for all small $x > 0$ and $1 + x \le e^{x}$ for all $x$. 

To bound the maximum type II error probability, first note that if $f \in \scF_{1n} := \{f: |\bar F(t_n) - \bar F^*(t_n)| > 2\epsilon_n\}$ then $\bbP^n_f(1 - T_n) \le \bbP^n_f(1 - T_{1n}) \le 2 e^{-2n\epsilon_n^2}$ by another application of Hoeffding's inequality. Next consider an $f \in \scT(t_n, \delta_n) \setminus \scF_{1n}$ with $\ti(f) < \alphahi$ and $|\ti(f) - \alpha^*| > 2^{4+\alphahi} \delta_n$. Let $n$ be large enough so that $\delta_n < 1/2$. It follows from the definition of $\scT(t, \delta)$ that $|\frac{\bar F(2t_n) }{ \bar F(t_n)} - 2^{-\ti(f)}| < 2^{1- \ti(f)}\delta_n < 2\delta_n$ and hence 
\[
\textstyle |\frac{\bar F(2t_n) }{\bar F(t_n)} - \frac{\bar F^*(2t_n) }{ \bar F^*(t_n)}| \ge 2^{-\max(\ti(f),\alpha^*)} \log(2)|\ti(f) - \alpha^*| - 4 \delta_n > 2\delta_n.
\] 
Consequently, $\bbP_f^n(1 - T_{2n}) \le 2\bbP^n_f[2e^{-2\bar S_n(t_n)\delta_n^2}] \le 2e^{-n\bar F_n(t_n) \delta_n^2}$. Since $f \not\in \scF_{1n}$, it follows that $\bar F(t_n) \ge \bar F^*(t_n) - 2\epsilon_n \ge \frac13\bar F^*(t_n)$ and hence $\bbP^n_f(1 - T_{2n}) \le 2e^{-n\epsilon_n^2}$. \end{proof}


\begin{proof}[Proof of Lemma 12]
Suppose $f = p_{\theta, \psi}$ with $\theta = (\alpha, \sigma) \in \Theta$, $\alpha \ge \alpha_1$, and $\psi = \scL(\omega)$, $\|\dot\omega\|_\infty \le D_n$. Denote $\phi = \log \psi$ and use Taylor's theorem to write $\bar F(y) = \psi(1)\bar G_\theta(y)\{1 - R_{\theta,\psi}(y)\}$ where $R_{\theta,\psi}(y) = \frac{\dot \psi(1 - u)}{2\psi(1)} \bar G_\theta(y) = \frac12 e^{-u\dot \phi(1 - v)}\dot \phi(1 - u) \bar G_\theta(y)$ for some $0 < v < u < \bar G_\theta(y)$. Notice that $\bar G_\theta(y) = (\alpha\sigma/y)^\alpha\{1 + r_\theta(y)\}$ with $|r_\theta(y)| < \alphahi^2\sigmahi/y \le \alphahi^2\sigmahi/t_n$ for all $y \ge t_n$ and consequently, $\bar G_\theta(y) \le c_1 t_n^{-\alpha} \le c_1 \tau_n / D_n$ for all $y \ge t_n$, for some fixed constant $c_1$. Since $\|\dot \phi\|_\infty = \|\dot \omega\|_\infty \le D_n$, it follows that for all large $n$, $|R_{\theta, \psi}(y)| \le \frac12 e^{c_1\tau_n}c_1 \tau_n \le 2c_1\tau_n$ for all $y \ge t_n$ and consequently, 
\[
\frac{y^\alpha\bar F(y)}{\zeta(f)} = \{1 + r_\theta(y)\}\{1 - R_{\theta,\psi}(y)\} = 1 + \tilde R_{\theta,\psi}(y)
\] 
with $|\tilde R_{\theta, \psi}(y)| \le 3 \max(|R_{\theta,\psi}(y)|, |r_\theta(y)|) \le B_1 \tau_n$ for all $y \ge t_n$, for some constant $B_1$. This concludes the proof since the choice of $B_1$ does not depend on $f$.  
\end{proof}

\newpage
\section*{Additional summary of numerical experiments}
\begin{figure*}[!h]
\centering
\includegraphics[width=\textwidth]{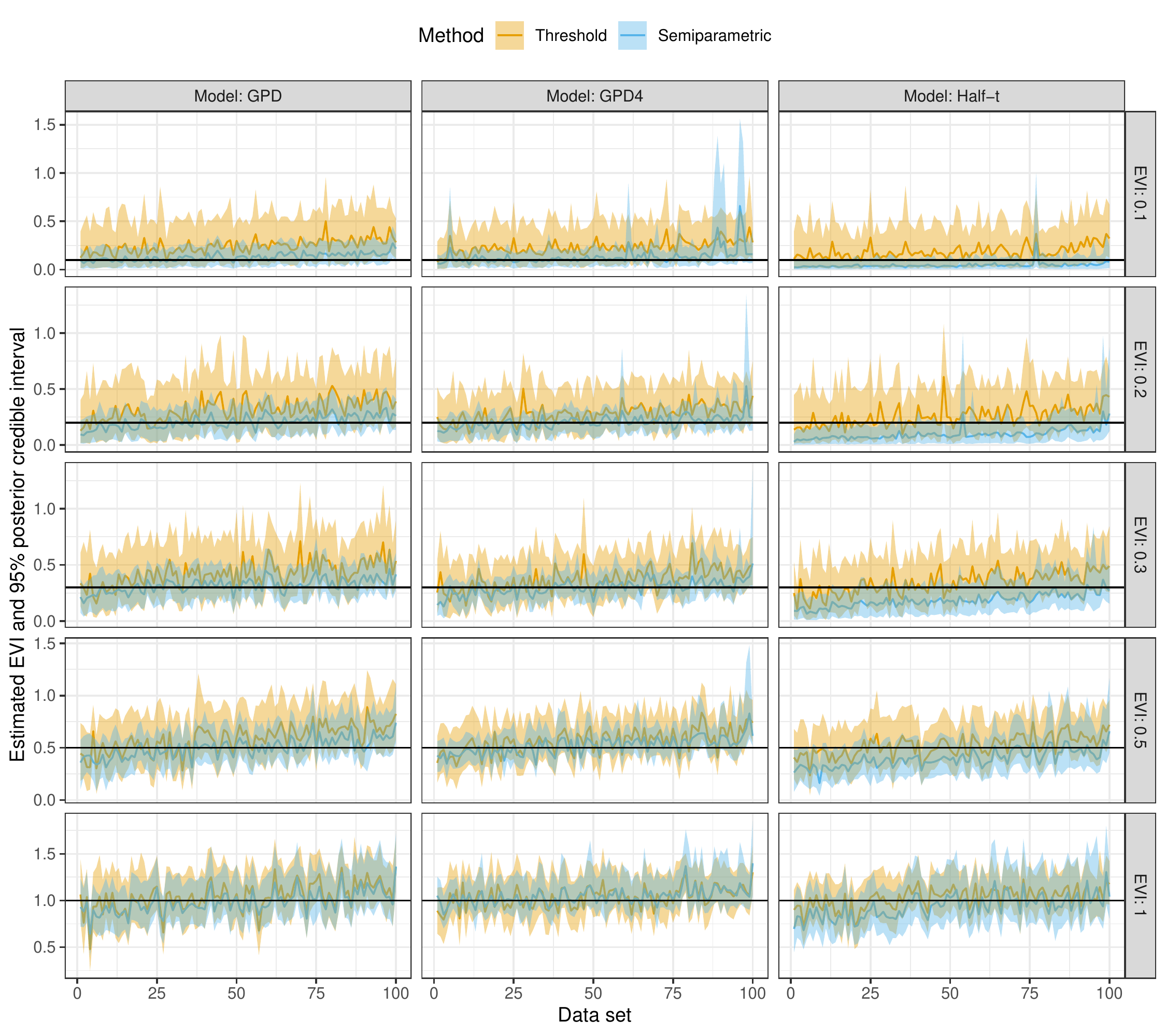}
\caption{A comparison of the 95\% posterior credible intervals for $\xi = \alpha^{-1}$ from the semiparametric and the thresholding methods. For each group, the 100 data sets are arranged in the ascending order of the maximum observation.}
\label{fig:intervals}
\end{figure*}
%
%
%
%
%
%
%
%
%

\section*{Codes}
R package `\texttt{sbde}' can be downloaded from \url{https://CRAN.R-project.org/package=sbde}. Follow the link \url{https://www2.stat.duke.edu/~st118/Codes-FortCollins/} to access R codes along with dataset and auxiliary codes required to reproduce Fort Collins precipitation analysis presented in this article. 
%
%
\newpage

\bibliographystyle{chicago}
\bibliography{MetaBib}
\end{document}